\documentclass[a4paper,USenglish,cleveref,autoref]{lipics-v2021}
\hideLIPIcs  %uncomment to remove references to LIPIcs series (logo, DOI, ...), e.g. when preparing a pre-final version to be uploaded to arXiv or another public repository

\graphicspath{{figures/}} %helpful if your graphic files are in another directory

\ifpdf
\hypersetup{
	pdftitle={Polyline Simplification under the Local Fr\'echet Distance has Almost-Quadratic Runtime in 2D},
	pdfauthor={S.\ Storandt and J.\ Zink}
}
\fi

\usepackage[disable]{todonotes}

\usepackage{amssymb}
\usepackage{thmtools}
\usepackage{color,colortbl}
\definecolor{lightgray}{gray}{0.85}
\usepackage{subcaption}
\usepackage{graphicx}
\usepackage{xspace}

\newcommand{\appmark}{$\star$}

\newcommand{\ohhat}{\ensuremath{\hat{\mathcal{O}}}}
\newcommand{\oh}{\ensuremath{\mathcal{O}}}

\newcommand{\hau}{Hausdorff\xspace}
\newcommand{\fre}{Fr\'echet\xspace}
\newcommand{\lone}{L$_1$\xspace}
\newcommand{\ltwo}{L$_2$\xspace}
\newcommand{\linf}{L$_{\infty}$\xspace}
\newcommand{\lp}{L$_p$\xspace}

\usepackage{mathtools}

\title{Polyline Simplification under the Local \fre Distance Has Almost-Quadratic Runtime in 2D}
\titlerunning{Polyline Simplification under the Local \fre Distance} %TODO optional, please use if title is longer than one line

\author{Sabine Storandt}{Universit\"at Konstanz, Konstanz, Germany}{sabine.storandt@uni-konstanz.de}{}{}
\author{Johannes Zink}{Universit\"at W\"urzburg, W\"urzburg, Germany}{zink@informatik.uni-wuerzburg.de}{https://orcid.org/0000-0002-7398-718X}{}

\authorrunning{S.\ Storandt and J.\ Zink} %TODO mandatory. First: Use abbreviated first/middle names. Second (only in severe cases): Use first author plus 'et al.'

\Copyright{Sabine Storandt and Johannes Zink} %TODO mandatory, please use full first names. LIPIcs license is "CC-BY";  http://creativecommons.org/licenses/by/3.0/

\ccsdesc[500]{Theory of computation~Computational geometry} %TODO mandatory: Please choose ACM 2012 classifications from https://dl.acm.org/ccs/ccs_flat.cfm 

\keywords{Polyline simplification, Local \fre distance, Exact algorithm, $p$-norm} %TODO mandatory; please add comma-separated list of keywords

\category{} %optional, e.g. invited paper

\relatedversion{} %optional, e.g. full version hosted on arXiv, HAL, or other respository/website
\nolinenumbers %uncomment to disable line numbering

\begin{document}
	
\maketitle

\begin{abstract}
Given a polyline on $n$ vertices, the polyline simplification problem
asks for a minimum-size subsequence
of these vertices defining a new polyline whose distance to the original polyline is at most a given threshold under some distance measure,
usually the local \hau or the local \fre distance.
Here, \emph{local} means that, for each line segment of the simplified polyline,
only the distance to the corresponding sub-curve in the original polyline is measured.

This minimization problem is polynomial-time solvable
(for the standard distance measures)
as first shown by Imai and Iri for the local \hau distance
by a cubic-time algorithm.
Later, Chan and Chin improved this running time to $\oh(n^2)$.
For the local \fre distance, the situation has been more intricate
for the last decades.
Doubtlessly, cubic running time is possible
as shown by Godau adjusting the Imai--Iri algorithm to the \fre metric.
This running time has been cited as state-of-the-art many times
for the last 30 years also in influential articles.
Very recently Buchin et al.\ [ESA '22] showed how to improve this running time
to $\oh(n^{5/2+ \varepsilon})$ for any $\varepsilon > 0$ as an application
of a sophisticated data structure to store polylines
for \fre distance queries.

But there are actually older techniques and concepts in literature capable of
solving this problem in almost-quadratic time.
Namely, Melkman and O'Rourke [Computational Morphology '88]
introduced a geometric data structure to solve
polyline simplification under the local \hau distance
in $\oh(n^2 \log n)$ time,
and Guibas, Hershberger, Mitchell, and Snoeyink
[Int.\ J.\ Comput.\ Geom.\ Appl.~'93]
considered polyline simplification under the \fre distance as \emph{ordered stabbing}
and provided an algorithm with a running time of $\oh(n^2 \log^2 n)$,
but they did not restrict the simplified polyline to use only vertices
of the original polyline.

We show that their techniques can be adjusted to solve polyline 
simplification under the local \fre distance in $\oh(n^2 \log n)$ time.
This algorithm may serve as a more efficient subroutine
for multiple other algorithms.
We provide a simple algorithm description as well as
rigorous proofs to substantiate this theorem.
We also investigate the geometric data structure introduced by
Melkman and O'Rourke, which we refer to as \emph{wavefront},
in more detail and feature some interesting properties.
As a result, we can prove that under the \lone and the \linf norm, the algorithm
can be significantly simplified and then only requires a running time of $\oh(n^2)$.
We also define a natural class of polylines where our algorithm
always achieves this running time also in the Euclidean norm \ltwo.
\end{abstract}

\section{Introduction}
Polyline simplification has a long history in computational geometry,
where it has also been known as \emph{polygonal approximation},
\emph{line generalization}, or \emph{$\varepsilon$-simplification}.
It owes its relevance~-- also beyond computational geometry~-- to a large variety of applications,
such as processing of vector graphics~\cite{Kreveld2020,wu2003non},
robotics~\cite{Nguyen2007,Duda1973},
trajectory clustering~\cite{brankovic2020k},
shape analysis~\cite{min2019sgm},
data compression~\cite{meratnia2004spatiotemporal},
curve fitting~\cite{Ramer1972},
and map visualization~\cite{ahmed2015frechet,Ahmed2012,Gruppi2015,Isenberg2013,Visvalingam1990}.
%Polyline simplification is an extensively studied topic due to its relevance
%to a variety of applications, such as trajectory and shape analysis \cite{min2019sgm},
%data compression \cite{meratnia2004spatiotemporal} or map visualization \cite{ahmed2015frechet}.
The task of polyline simplification 
is to replace a given polyline on $n$ vertices with a
minimum-size subsequence of its vertices while
ensuring that the input and the output polyline are  sufficiently similar.
The similarity is governed by a given distance threshold $\delta$.
Line segments between vertices in the output polyline are called \emph{shortcuts}.
To determine the similarity of the input and output polyline,
the \hau and the \fre distance are the most commonly used measures.
Both can be applied either globally or locally.
In the global version, the distance between the entire input and output
polyline is measured and must not exceed~$\delta$.
In the local version, the distance between each shortcut and the part of
the input polyline it bridges must not exceed $\delta$. 

For many applications, local similarity is more sensible and intuitive.
Simplifications with global similarity have only been studied recently.
For the global (undirected) \hau distance, computing a simplified polyline with the
smallest number of shortcuts yields an NP-hard problem \cite{Kreveld2020}.
For the global \fre distance, an $\oh(n^3)$ time algorithm was designed
by Bringmann and Chaudhury~\cite{bringmann2019polyline}.
For the more extensively studied and more classical problem
of simplification with the local \hau distance,
the Imai--Iri algorithm~\cite{Imai1988}, published in 1988, guarantees
a running time of $\oh(n^3)$ by reducing the simplification
problem to a graph problem.
Essentially, the algorithm constructs a graph where the polyline vertices are the nodes
and where there is an edge between a pair of vertices if they can be connected with
a shortcut respecting the $\delta$ distance bound.
For the local \hau distance, Melkman and O'Rourke~\cite{melkman1988polygonal} showed
already in 1988 that the Imai--Iri algorithm can be improved to run in
$\oh(n^2 \log n)$ time by making the graph construction phase more efficient.
In 1996, Chan and Chin~\cite{Chan1996} further reduced the running time to $\oh(n^2)$.

For the local \fre distance, though, the cubic running time of the Imai--Iri algorithm, which was shown by Godau~\cite{Godau91} in 1991,
was a longstanding bound and used as a building block or reference also in recent publications~\cite{bringmann2019polyline, Kreveld2020}.
Agarwal, Har-Peled, Mustafa and Wang~\cite{Agarwal2005} explicitly posed the problem of whether there exists a subcubic algorithm for polyline simplification under the local \fre distance as an open question in 2005.
The question was answered positively very recently
by Buchin, van der Hoog, Ophelders, Schlipf, Silveira, and Staals~\cite{buchin2022efficient}.
They describe a data structure that outputs the \fre distance between any line segment and any subpolyline of a preprocessed input polyline in $\oh(\sqrt n \log^2 n)$ time.
Using this data structure to check whether there is a valid shortcut for every vertex pair, a polyline simplified optimally can be computed in $\oh(n^{5/2+\varepsilon})$ time (and space) for any $\varepsilon > 0$.
We remark that this data structure is quite sophisticated and actually more powerful than required for polyline simplification.
To check whether there is a valid shortcut between a vertex pair with respect to the \fre distance, it suffices to be able to decide whether the distance of the shortcut to its subpolyline is at most $\delta$.
However, the data structure always returns the exact distance value and it can accomplish this for arbitrary line segments and not only potential shortcuts.

\subsection{Related Work}
The most practically relevant setting for polyline simplification is to consider two-dimensional input curves  in the Euclidean plane (i.e., under the \ltwo norm).
However, the problem was also studied in higher dimensions $d > 2$ and under different norms.
%which are also relevant from a more general viewpoint.

The $\oh(n^3)$ time algorithm by Imai, Iri, and Godau~\cite{Imai1988, Godau91}
for polyline simplification under the local \hau and \fre distance
as well as the $\oh(n^3)$ time algorithm by Bringmann and Chaudhury~\cite{bringmann2019polyline}
for the global \fre distance can be generalized to work in
$\mathbb{R}^{d \geq 2}$  with the running time only increasing
by a polynomial factor in $d$.
For the local \hau distance, Chan and Chin \cite{Chan1996} showed
that the Imai--Iri algorithm can be improved to run in $\oh(n^2)$ time
for \lone, \ltwo and \linf (the concept can also be applied to any
L$_{p \in (1,\infty)}$ up to possible numerical issues that are further
discussed in \cref{sec:allp}).
Furthermore Barequet et al.~\cite{barequet2002efficiently} proposed
an $\oh(n^2 \log n)$ time algorithm for the local \hau distance
under the \ltwo norm %which works
in $\mathbb{R}^3$, as well as an $\oh(d2^dn^2)$ algorithm for \lone and an $\oh(d^2n^2)$ algorithm for \linf.

\begin{table}
	\centering
	\begin{tabular}{l|l|l|l|l}
		& Conditional  &  & \multicolumn{2}{c}{\fre}\\
		& Lower Bound&Local \hau & Local & Global \\	
		\hline
		\rowcolor{lightgray}
		\lone		& $\ohhat(n^{3-\varepsilon})$~\cite{bringmann2019polyline}	& $\ohhat(2^{d}n^2)$ ~\cite{barequet2002efficiently}  & \textcolor{blue}{$\oh(n^2)~d=2$ [Thm.~\ref{clm:lonelinf-n2}]}	 & $\ohhat(n^3)$~\cite{bringmann2019polyline}\\
		\rowcolor{lightgray}
		& & $\ohhat(n^3)$~\cite{Imai1988}& $ \ohhat(n^3)$~\cite{Godau91}	 &\\ 
		L$_{\stackrel{p \in (1,\infty)}{_{p\neq 2}}}$		& $\ohhat(n^{3-\varepsilon})$~\cite{bringmann2019polyline}	&$\oh(n^2 )~d=2$~\cite{Chan1996}  & \textcolor{blue}{$\oh(n^2 \log n)~d=2$ [Cor.~\ref{clm:lp-n2logn}]}	 & $\ohhat(n^3)$~\cite{bringmann2019polyline}\\
		& & $\ohhat(n^3)$~\cite{Imai1988}& $ \ohhat(n^3)$~\cite{Godau91}	 &\\\rowcolor{lightgray}
		\ltwo		& $\ohhat(n^{2-\varepsilon})$~\cite{buchin2016fine}	&$\oh(n^2)~d=2$~\cite{Chan1996}  &	\textcolor{blue}{$\oh(n^2 \log n)~d=2$ [Thm.~\ref{clm:ltwo-n2logn}]}	& $\ohhat(n^3)$~\cite{bringmann2019polyline}\\
		\rowcolor{lightgray}&	&$\oh(n^2 \log n)$~$d=3$~\cite{barequet2002efficiently}	& $\oh(n^{5/2+\varepsilon})$~$d=2$ \cite{buchin2022efficient} &	\\
		\rowcolor{lightgray}&	& $\ohhat(n^3)$~\cite{Imai1988} &  $\ohhat(n^3)$~\cite{Godau91} &	\\
		%\rowcolor{lightgray}&	& $\ohhat(n^{3-\Omega(1/d)})$~\cite{barequet2002efficiently}	& $ \ohhat(n^3)$~\cite{Godau91}&	\\
		%L$_{p \in (2,\infty)}$ & $\ohhat(n^{3-\varepsilon})$~\cite{bringmann2019polyline}& $\oh(n^2)~d=2$~\cite{Chan1996} 	& \textcolor{blue}{$\oh(n^2 \log n)~d=2$ [\Cref{clm:lp-n2logn}]} 	& $\ohhat(n^3)$~\cite{bringmann2019polyline}\\
		%& & $ \ohhat(n^3)$~\cite{Imai1988}& $\ohhat(n^3)$~\cite{Godau91}	&\\\rowcolor{lightgray}
		\linf	&$\ohhat(n^{2-\varepsilon})$~\cite{buchin2016fine}& $\ohhat(n^2)$~\cite{barequet2002efficiently}	&\textcolor{blue}{$\oh(n^2)~d=2$ [Thm.~\ref{clm:lonelinf-n2}]} 	& $\ohhat(n^3)$~\cite{bringmann2019polyline}\\
		%\rowcolor{lightgray}
		& &  & $ \ohhat(n^3)$~\cite{Godau91}	 &\\
	\end{tabular}
	\caption{Conditional lower bounds (presented as running times that are excluded)
		and algorithmic upper bounds for polyline simplification in $\mathbb{R}^d$
		under different similarity measures and \lp metrics.
		Here, $n$ is the number of vertices, $d$ the number of dimensions
		and $\varepsilon$ is any constant $> 0$.
		The $\ohhat$-notation hides
		polynomial factors in $d$. Blue entries mark the results
		presented in this manuscript.}\label{tab:overview}
	\vspace{-3ex}
\end{table}

Bringmann and Chaudhury~\cite{bringmann2019polyline} have also proven a conditional 
lower bound for simplification in $\mathbb{R}^d$ under the local \hau distance
as well as under the local and global \fre distance.
More precisely, for \lp with $p \in [1,\infty), p\neq 2$, algorithms with
a running time subcubic in $n$ and polynomial in $d$ were ruled out
(unless the $\forall\forall\exists$-OV hypothesis fails).
For large dimensions $d$, the algorithmic upper bounds of $\oh(n^3 \cdot \mathrm{poly}(d))$
for polyline simplification discussed above are hence tight.
However, the lower bound  still allows the existence of simplification algorithms
with a running time in $\oh(n^k \cdot \exp(d))$ with $k<3$.
Hence, for small values of $d$ (which are of high practical relevance),
faster algorithms are possible, as  evidenced by the $\oh(d2^dn^2)$ time algorithm
for the local \hau distance under the \lone norm \cite{barequet2002efficiently}.
For \ltwo and \linf, the best currently known conditional lower bound for
the three similarity measures, local \hau distance, local \fre distance, and
global \fre distance, was proven by Buchin et al.~\cite{buchin2016fine}.
It rules out algorithms with a subquadratic running time in $n$ and polynomial
running time in $d$ (unless SETH fails).
Here again, better running times for  simplification problems in a low-dimensional
space with $d \in o(\log n)$ are still possible. 
\cref{tab:overview} provides an overview of known lower and
upper bounds for optimal polyline simplification.

As the cubic running time  of the Imai--Iri algorithm
and the quadratic running time of the Chan--Chin algorithm
may be prohibitive for processing long polylines even for $d=2$,
heuristics and approximation algorithms have been investigated.
The Douglas--Peucker algorithm \cite{Douglas1973},
one of the most simple and widely used heuristics, computes a simplified
polyline under the local \hau distance in $\oh(n \log n)$ time~\cite{Hershberger1992}
and under the local \fre distance in $\oh(n^2)$ time~\cite{Kreveld2020}~--
but without any guarantee regarding the solution size.
Agarwal et al.~\cite{Agarwal2005} presented an approximation algorithm with a running time of
$\oh(n \log n)$ that works for any \lp norm and
generalizes to $\mathbb{R}^d$.
It computes a simplification under the local \fre distance for $\delta$,
where the simplification size does not exceed
the optimal simplification size for $\delta/2$.
There are other heuristics neither based on the \hau nor the \fre distance
like the algorithm by Visvalingam and Whyatt~\cite{Visvalingam1993},
which measures the importance of a vertex by the triangular area it adds. 

There are also variants of the \fre distance which allow for faster polyline simplification.
For example, polyline simplification under the discrete \fre distance
(where only the distance between the vertices but not the points on
the line segments in between matters) can be solved to optimality
in $\oh(n^2)$ time \cite{bereg2008simplifying}.
However, the discrete \fre distance heavily depends on the density of vertices on the polyline,
and hence for many applications the
continuous \fre distance studied in this paper constitutes a more meaningful measure.

The problem variant where the requirement is dropped that all vertices of the simplification must be vertices of the input polyline is called a \emph{weak} simplification.
Guibas et al.~\cite{Guibas1993} showed that an optimal
weak simplification under the (global) \fre distance %with distance threshold~$\delta$
can be computed in $\oh(n^2 \log^2 n)$ time.
Later Agarwal et al.~\cite{Agarwal2005} gave an $\oh(n \log n)$-time
approximation algorithm for a weak simplification violating the distance threshold~$\delta$
by a factor of at most~8; see also an overview by Van de Kerkhof et al.~\cite{van2019global}.

Regarding our techniques, we remark that the concept of a \emph{wavefront}
being comprised of (circular) arcs is well-established in computational geometry.
Mitchell, Mount, and Papadimitriou~\cite{Mitchell1987} have introduced
the \emph{continuous Dijkstra} method in 1987, where a wavefront is expanded%
\footnote{This expansion is realized step-wise by discrete events.}
along a surface of a polyhedron to allow shortest path computations.
This method was also used by Hershberger and Suri~\cite{Hershberger1999}
to compute a shortest path in the plane given a set of polygonal obstacles.
Also, many sweep-line algorithms maintain a kind of a wavefront.
For example, the famous algorithm by Fortune~\cite{Fortune1987}
for computing a Voronoi diagram maintains a wavefront made up
of hyperbolic curves.
Another approach computes a (weighted) Voronoi diagram directly by
a wavefront expanding around the input points~\cite{Held2020}.
Apart from a few similarities regarding the computation of
intersection points of (circular) arcs, line segments, etc.,
we use a different type of wavefront.
In all of the aforementioned examples, the wavefront is a \emph{kinetic} data structure
where things move or expand continuously (although it suffices to consider a few discrete events).
In contrast, the structure by Melkman and O'Rourke~\cite{melkman1988polygonal}
that we call a wavefront is more static and simple:
it is a collection of (circular) arcs that we 
update iteratively with a new unit circle.
There is no expansion of existing arcs or ``time in between''.

\subsection{Contribution}
We present an algorithm for polyline simplification under the local \fre distance
in two dimensions and for several \lp norms with a (near-)quadratic running time;
see \cref{sec:algorithm}.

Our algorithm heavily builds upon the Melkman--O'Rourke algorithm~\cite{melkman1988polygonal}.
They exploit the geometric properties of the local \hau distance using cone-shaped \emph{wedges} and a \emph{wavefront} to accelerate the shortcut graph construction.
We adapt both of these concepts to the local \fre distance.
We carefully study the properties of the resulting wavefront
and explain how to maintain and efficiently update a wavefront data structure
which at its core is a simple balanced binary search tree.
As our main result,
we prove that the asymptotic running time of the Melkman--O'Rourke algorithm does not increase with our modifications;
%and hence optimal simplifications under the local \fre distance can be
%computed in $\oh(n^2 \log n)$ time using $\oh(n)$ space.
see \cref{sec:wavefront}.

\begin{restatable}{theorem}{ltwotheorem}
	\label{clm:ltwo-n2logn}
	A two-dimensional $n$-vertex polyline can be simplified optimally
	under the local \fre distance in the~\ltwo norm (the Euclidean norm)
	in $\oh(n^2 \log n)$ time and $\oh(n)$~space.
\end{restatable}

This is a large improvement compared to the cubic running time
by Imai and Iri and by Godau and also to the best currently claimed
running time bound of $\oh(n^{5/2+\varepsilon})$~\cite{buchin2022efficient},
compared to which it is also simpler.
It is also faster than the  $\oh(n^2 \log^2 n)$ running time of
the weak simplification algorithm by Guibas et al.\ 
(\cite{Guibas1993}, Theorem~14) by a logarithmic factor.
However, we remark that parts of their algorithm
(Def.~4, Theorem~7, Lemma~8, Lemma~9) can be used
to obtain the same result as described in their article
and we partially re-use their techniques.
Yet, their procedure is more complicated since it maintains
more geometric information only needed for weak simplifications.
Our algorithm is hence more straight-forward for the setting of 
polyline simplification under the local \fre distance,
which makes it conceptually easier to understand.

Furthermore, we show that under the \lone and \linf norm,
the wavefront has constant complexity which improves the running time to $\oh(n^2)$.
We  argue that for
a natural class of polylines, a quadratic running time can be achieved as well.
To this end, we introduce the concept of \emph{$\nu$-light} polylines;
see \cref{sec:small-wave-fronts}.

\begin{restatable}{theorem}{lonelinftheorem}
	\label{clm:lonelinf-n2}
	A two-dimensional $n$-vertex polyline can be simplified optimally under the
	local \fre distance in the \lone and \linf norm in $\oh(n^2)$ time and $\oh(n)$ space.
\end{restatable}

Besides this new application for the local \fre distance, investigating
the structure and implementation of \emph{wedges} and \emph{wavefronts} is of independent interest
and may also help to understand better the work
by Melkman and O'Rourke~\cite{melkman1988polygonal} from 1988
%which leaves some details
and the algorithm by Guibas et al.~\cite{Guibas1993},
who both employ these data structures but give little detail
on its structural properties and on how to perform operations with the wavefront.
To this end, we have an extensive appendix with detailed proofs,
which we have separated from the main part to keep the algorithm
description more compact and understandable on its own,
while the interested reader can find the additional content in the appendix.
Statements whose proofs can be found
in the appendix are marked with ``\appmark''.

\section{Preliminaries}
We start with some basic definitions in the field of polyline simplification,
then recapitulate the important ingredients of the polyline simplification algorithms
by Imai and Iri~\cite{Imai1988}, Melkman and O'Rourke~\cite{melkman1988polygonal},
as well as Guibas et al.~\cite{Guibas1993},
and finally specify the notation we use thereupon.

\subsection{Basic Definitions}
A \emph{polyline}
is a series of line segments that are defined
by a sequence of $d$-dimensional points $L = \langle p_1, p_2, \dots, p_n \rangle$,
which we call \emph{vertices}.
By $n$, we denote the \emph{length} of a polyline.
For $1 \le i \le j \le n$,
we let $L[p_i, p_j] := \langle p_i, p_{i+1}, \dots, p_j \rangle$,
that is, the subpolyline of $L$ starting
at vertex $p_i$, ending at vertex $p_j$, and including all vertices in between in order.
The continuous (but not smooth) curve induced by the vertices of a polyline $L$ is denoted
as $c_{L} \colon [1, n] \rightarrow \mathbb{R}^d$
with $c_{L} \colon x \mapsto (\lfloor x \rfloor + 1 - x) p_{\lfloor x \rfloor}
+ (x - \lfloor x \rfloor) p_{\lceil x \rceil}$.
The polyline simplification problem is  defined as follows;
for an example see \cref{fig:example-simplification}.

\begin{figure}
	\centering
	\includegraphics[]{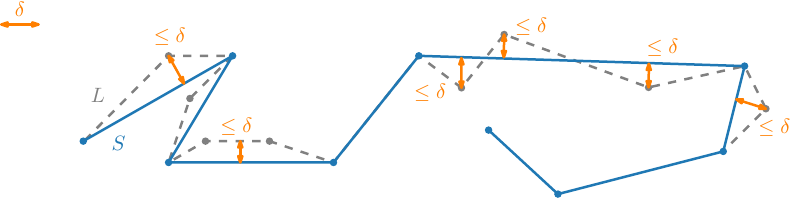}
	\caption{A polyline~$L$ (dashed gray) and a simplification~$S$ of $L$ (solid blue) according to a distance measure with distance parameter $\delta$.}
	\label{fig:example-simplification}
\end{figure}

\begin{definition}[Polyline Simplification]
	Given a polyline $L = \langle p_1, p_2, \dots, p_n \rangle$,
	a distance measure~$d_{X}$ for determining the distance between two polylines,
	and a distance threshold parameter~$\delta$,
	the objective is to obtain a minimum-size subsequence $S$ of $L$
	such that %\dots
	%	\begin{itemize}
	%		\item
	$p_1, p_n \in S$, and
	%		\item
	$d_{X}(L, S) \le \delta$.
	%	\end{itemize}
	We refer to $S$ as a \emph{simplification} of the (original) polyline $L$.
\end{definition}

Next, we discuss typical candidates for such a distance measure~$d_{X}$,
namely the \hau and the \fre distance in their local and their global variant.
\begin{definition}[\hau Distance]
	Given two polylines $L = \langle p_1, \dots, p_n \rangle$
	and $L' = \langle q_1, \dots, q_m \rangle$,
	the \emph{(undirected) \hau distance}~$d_\textnormal{H}(L, L')$ is defined as
	\begin{equation*}
	d_{\textnormal{H}}(L, L') := \max \left\{
	\adjustlimits\sup_{p \in c_{L}} \inf_{q \in c_{L'}} d(p, q),
	\adjustlimits\sup_{q \in c_{L'}} \inf_{p \in c_{L}} d(p, q)
	\right\} \, ,
	\end{equation*}
	where
	$\sup$ is the supremum, $\inf$ is the infimum and
	$d(p,q)$ is the distance between the points $p$ and $q$ under some norm
	(e.g., the \ltwo norm).
	%	(typically the Euclidean distance).
\end{definition}

An often raised criticism concerning the use of the \hau distance
is that it does not reflect the similarity of the courses of two polylines.
In contrast, the \fre distance measures the maximum distance
between two polylines while traversing them in parallel and is
therefore often regarded as the better suited measure
for polyline similarity.

\begin{definition}[\fre Distance]
	Given two polylines $L = \langle p_1, \dots, p_n \rangle$
	and $L' = \langle q_1, \dots, q_m \rangle$,
	the \emph{\fre distance}~$d_\textnormal{F}(L, L')$ is defined as
	\begin{equation*}
	d_\textnormal{F}(L, L') := \adjustlimits\inf_{\alpha, \beta} \max_{t \in [0, 1]}
	d(c_{L}(\alpha(t)), c_{L'}(\beta(t))),
	\end{equation*}
	where $\alpha \colon [0, 1] \rightarrow [1,n]$ and $\beta \colon [0, 1] \rightarrow [1, m]$
	are continuous and non-decreasing functions such that
	$\alpha(0) = \beta(0) = 1$, $\alpha(1) = n$, and $\beta(1) = m$.
\end{definition}

Observe that the \hau distance is a lower bound for the \fre distance.
Later, when considering the \fre distance, we may say
that the distance threshold $\delta$ is respected or exceeded already in the \hau distance
since not exceeding $\delta$ for the \hau distance is a necessary
condition to not exceed $\delta$ for the \fre distance.

Conventionally, in the context of polyline simplification,
the \emph{local} \hau and \emph{local} \fre distance is used,
which only measures the \hau or \fre distance between
a line segment $\langle p_i,p_j \rangle$ of the simplification and its corresponding subpolyline $L[p_i,p_j]$ in the original polyline.
Similar to the advantage of the \fre distance over the \hau distance,
the advantage of a local over a global distance measure is
that the we compare only related parts of a polyline and its simplification.
Hence, using the local \fre distance for polyline simplification is an arguably sensible choice.
%We define only the local \fre distance~--
%the local \hau distance is defined analogously.%
%\begin{definition}[Local \fre Distance]
%	Given a polyline $L = (p_{1}, p_{2}, \dots, p_{n})$
%	and a simplification $S = (p_1 = p_{s_1}, p_{s_2}, \dots, p_{s_{|S|}} = p_n)$ of $L$,
%	the \emph{local \fre distance}
%%	$d_\textnormal{lF}(S, L)$ is defined as
%	\begin{equation*}
%	d_\textnormal{lF}(S, L) := \max_{i \in {1, \dots |S| - 1}}
%	d_\textnormal{F}((p_{s_i}, p_{s_{i + 1}}), L[p_{s_i},p_{s_{i + 1}}]) \, ,
%	\end{equation*}
%	where $\langle p_{s_i}, p_{s_{i + 1}} \rangle$ is the polyline of length two
%	(i.e., the line segment) from $p_{s_i}$ to $p_{s_{i + 1}}$
%	and $L[p_{s_i},p_{s_{i + 1}}]$ is the (sub)polyline
%	we obtained by taking the substring from $p_{s_i}$ to $p_{s_{i + 1}}$ of $L$.
%\end{definition}

When using a local distance measure,
we can tell for each pair of vertices~$\langle p_i, p_j \rangle$ (for $1 \le i < j \le n$) in the original polyline independently
whether a simplification may contain the line segment $\langle p_i, p_j \rangle$ or not
by only considering the distance between the line segment~$\langle p_i, p_j \rangle$ and its
corresponding subpolyline.
When considering such a pair $\langle p_i, p_j \rangle$ 
as a line segment for a simplification, we call it a \emph{shortcut}.
If the %\hau/\fre
distance between a segment $\langle p_i, p_j \rangle$  and its corresponding subpolyline
does not exceed the distance threshold~$\delta$, we call it a \emph{valid} shortcut.
Note that trivially $\langle p_i, p_{i + 1} \rangle$ is always a valid shortcut
for any $i \in \{1, \dots, n-1\}$.

\begin{figure}
	\centering
	\includegraphics[]{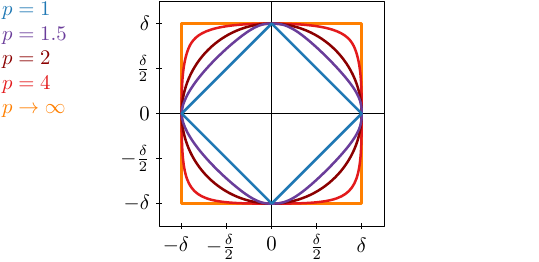}
	\caption{Unit circles in \lp norms for selected values of $p$.
	As their radius, we use $\delta$ instead of 1.}
	\label{fig:norm}
\end{figure}

In the definition of the \hau and \fre distance, we can choose how
the distance between two $d$-dimensional points, or vectors, is determined.
Typically, a vector norm is used for this purpose.
For $p \in [1, \infty)$, the \lp norm of a vector $x \in \mathbb{R}^d$ is
defined as $\|x\|_p := \left(\sum_{i=1}^d \left|x_i\right|^p\right)^{1/p}$.
For $p=1$, it is called the Manhattan norm, for $p=2$, the Euclidean norm.
For $p \rightarrow \infty$, \linf is called the maximum norm and
it is defined as $\max_{i=1,\dots,n} |x_i|$.
The unit sphere $S^d_p$ is the set of points in $\mathbb{R}^d$ within
unit distance to the origin. While this unit is conventionally set to $1$,
we use $\delta$ here instead as this allows for easier integration
to polyline simplification with error bound $\delta$.
We hence define $S_p^d := \{x \in \mathbb{R}^d|~ \|x\|_p \le \delta\}$.
For $d=2$, $S_p^2$ is also called the \emph{unit circle} in \lp;
we illustrate some examples in \cref{fig:norm}.
For \lone and \linf, it actually forms a square with side length
$\sqrt 2 \delta$ and $2\delta$, respectively.
For \ltwo, it is indeed a circle with radius $\delta$.
For $p$  between $2$ and $\infty$, it forms a supercircle which for
larger $p$ resembles more and more a square.
We refer to 
a contiguous subset of the boundary of a unit circle in \lp as an \emph{arc}.
%We use them to describe the regions in which valid shortcut endpoints may be located.

\subsection{Imai--Iri Algorithm% (Simplification under the Local \hau/\fre Distance)
}
\label{sec:imai-iri}
\label{fasterpolyline:sec:imai-iri}
Given an $n$-vertex polyline~$L = \langle p_1,\dots, p_n \rangle$,
the polyline simplification algorithm by Imai and Iri~\cite{Imai1988} proceeds in two phases.
In the first phase, the \emph{shortcut graph} is constructed.
This graph  has a node for each vertex of~$L$ and it has an edge between two nodes
if and only if there is a valid shortcut between the corresponding two vertices of~$L$.
For the \hau and the \fre distance,
it can be checked in $\oh(n)$ time whether the distance
between a line segment and a polyline having $\oh(n)$ vertices exceeds $\delta$~\cite{Alt1995}.
Hence, the total running time of the first phase amounts to $\oh(n^3)$.
In the second phase, a shortest path from the first node $p_1$
to the last node $p_n$  is computed in the shortcut graph,
which can be accomplished in $\oh(n^2)$ time.
In a naive implementation, the space consumption is in $\oh(n^2)$.
However, it is not necessary to first construct the full shortcut graph and to compute the shortest path subsequently.
Instead, the space consumption can be reduced to $\oh(n)$ by interleaving the two phases as follows:
For $p_i$, the shortest path distance $d_i$ from $p_i$ to $p_n$ via shortcuts
can be computed in linear time by considering all valid shortcuts $\langle p_i, p_j \rangle$
to vertices $p_j$ with $j > i$ and setting $d_i = 1 +\min_{\langle p_i,p_j \rangle} d_j$.
Hence, if the vertices are traversed in reverse order,
only the distance values for already processed vertices
and the shortcuts of the currently considered vertex
need to be kept in memory to compute the correct solution without
increasing the asymptotic running time.

\subsection{Melkman--O'Rourke Algorithm% (Simplification under the Local \hau Distance)
}
Since in the Imai--Iri algorithm the construction of the shortcut graph dominates the runtime,
accelerating this first phase also leads to an overall improvement.
Melkman and O'Rourke~\cite{melkman1988polygonal} introduced a faster technique
to compute the shortcut graph for the local \hau distance.
Starting once at each vertex $p_i$ for $i \in \{1, \dots, n\}$,
they traverse the rest of the polyline vertex by vertex in $\oh(n \log n)$ time
to determine all valid shortcuts originating at~$p_i$.

To this end, they maintain a cone-shaped region called \emph{wedge}
in which all valid shortcuts are required to lie.
When traversing the polyline, the wedge may become narrower iteratively.
Moreover, they maintain a \emph{wavefront},%
\footnote{%
	Melkman and O'Rourke~\cite{melkman1988polygonal} use the term
	\emph{frontier} instead of wavefront.
	Within the cone, they only call the region on
	the other side of the frontier \emph{wedge}
	and they call the associated data structure
	\emph{wedge data structure.}
	Our notation to call the whole cone \emph{wedge}
	is in line with the algorithm by Chan and Chin~\cite{Chan1996}.
}
which is a sequence of circular arcs of unit circles.
The wavefront subdivides the wedge
into two regions~-- a valid shortcut~$\langle p_i, p_j \rangle$ has the endpoint~$p_j$
in the region not containing~$p_i$.
In other words, a valid shortcut needs to cross the wavefront.
%We call this region \emph{valid region}.
The wavefront has size in $\oh(n)$ and is stored in a balanced search tree (in the original article an augmented 2-3-tree) 
such that querying and updating operations can be performed in amortized $\oh(\log n)$ time.

%For $\langle p_i p_j \rangle$ to be a valid shortcut,
%$p_j$ needs to lie within the wedge and behind the wavefront.
Containment in the wedge can be checked in constant time
and the position of a vertex relative to the wavefront can be determined in $\oh(\log n)$ time.
Updating the wedge %from $W_{i, j - 1}$ to $W_{i, j}$
can be done in constant time.
Updating the wavefront may involve adding an arc and removing several arcs.
Here, the crucial observation~\cite{melkman1988polygonal} is that the order of arcs on the wavefront
is reverse to the order of the corresponding unit circle centers~--
all with respect to the angle around~$p_i$.
This allows for binary search in $\oh(\log n)$ time
to locate a new arc within the wavefront.
Although a linear number of arcs may be removed from the wavefront
in a single step, over all steps any arc is removed at most once.
Amortized, this results in a running time of $\oh(n \log n)$ per
starting vertex $p_i$ and $\oh(n^2 \log n)$ in total.

\subsection{Algorithm by Guibas, Hershberger, Mitchell, and Snoeyink% Guibas et al.\ (Weak Simplification under the \fre Distance)
}
\label{sec:guibas}
Guibas et al.~\cite{Guibas1993} study weak polyline simplification.
There, given an $n$-vertex polyline $L=\langle p_1, \dots,p_n \rangle$ and a distance threshold $\delta$,
the objective is to compute any polyline $S=\langle q_1,\dots,q_m \rangle$ of smallest possible length~$m$
that hits all unit circles around the vertices in $L$ in the given order,
which they call \emph{ordered stabbing}.
To additionally have \fre distance at most~$\delta$ between $L$ and $S$,
each vertex $q_j$ of $S$ needs to be in distance $\le \delta$ to some point of $c_L$.
They describe a 2-approximation algorithm running in $\oh(n^2 \log n)$ time
and a dynamic program solving this problem exactly in $\oh(n^2 \log^2 n)$ time.

Both algorithms essentially rely on a subroutine
to decide whether there exists a line~$\ell$ (a \emph{stabbing line})
that intersects a given set of $n$ ordered unit circles $\langle C_1, \dots, C_n \rangle$
such that $\ell$ hits some points $\langle r_1, \dots, r_n \rangle$ with $r_i \in C_i$ for $i \in \{1, \dots, n\}$
in order (\cite{Guibas1993}, Def.~4).
This subroutine runs in $\oh(n \log n)$ time (\cite{Guibas1993}, Lemma~9).
It is based on an algorithm computing iteratively two hulls and
two limiting lines through the unit circles that describe
all stabbing lines (\cite{Guibas1993}, Algorithm~1).
They also maintain the wavefront as described in the Melkman--O'Rourke algorithm.
However, they add an update step
to ensure that the stabbing line respects the order of the unit circles\footnote{%
This is necessary because here the \fre distance is considered,
while Melkman and O'Rourke only considered the Hausdorff distance.}.
We use conceptually the same update step and explain it in more detail in
\cref{sec:algorithm-outline,sec:wavefront-maintenance}.

To compute an optimal weak simplification $S$ with the dynamic program,
this subroutine is called once per unit circle induced by vertices in $L$. 
Guibas et al.\ remark that the wavefront might have non-constant complexity.
Hence they refrain from storing it explicitly.
They only keep the wedge and support information in memory and construct
the remaining information necessary to perform the update steps on demand.
This further complicates the algorithm and adds a logarithmic factor
per unit circle to the overall running time, which then is in $\oh(n^2 \log^2 n)$.

%By analyzing the wavefront more thoroughly, we will show that it is possible  as well as simpler and  more efficient to explicitly store the wavefront and to perform updates directly on the wavefront data structure.
%JZ: I don't know if we can argue that way here because we are computing a simpler type of simplification (their turning points are not necessarily vertices)

\subsection{Definitions and Notation}

\begin{figure}
	\centering
	\begin{subfigure}[t]{\linewidth}
		\centering
		\includegraphics[page=1, trim = 0 0 0 0, clip]{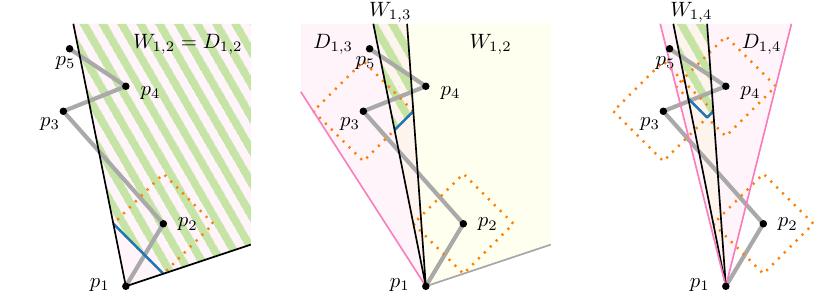}
		\caption{\lone norm: the unit circles are squares of side length~$\sqrt 2\delta$
			%whose boundary is
			rotated by 45~degrees w.r.t.\ the main axes.}
		\label{fig:wedges-and-wave-fronts-L1}
	\end{subfigure}
	
	\medskip
	
	\begin{subfigure}[t]{\linewidth}
		\centering
		\includegraphics[page=2, trim = 0 0 0 0, clip]{wedges-and-wave-fronts}
		\caption{\ltwo norm: the unit circles are circles of radius~$\delta$.
			The wavefront consists of $\oh(n)$ circular arcs.}
		\label{fig:wedges-and-wave-fronts-L2}
	\end{subfigure}
	
	\medskip
	
	\begin{subfigure}[t]{\linewidth}
		\centering
		\includegraphics[page=3, trim = 0 0 0 0, clip]{wedges-and-wave-fronts}
		\caption{\linf norm: the unit circles are squares of side length~$2 \delta$
			whose boundary is parallel to the main axes.
		}
		\label{fig:wedges-and-wave-fronts-Linf}
	\end{subfigure}
	\caption{Iterative construction of the wedge in the \lone, \ltwo and \linf norm:
		%		From left to right,
		The local wedges $D_{1,2}$, $D_{1,3}$, and $D_{1,4}$ are visualized in pink.
		Here, the intersections of local wedges define the wedges
		$W_{1,2}$, $W_{1,3}$, and $W_{1,4}$. %, respectively.
		The wavefront is a sequence of
		unit circle arcs and is depicted in blue.
		Within the wedge and above the wavefront, there is the valid region (depicted in hatched green).
		This is the area, where a subsequent vertex~$p_j$ needs to lie if there is a valid shortcut~$\langle p_1, p_j \rangle$.
		For example, $\langle p_1, p_5 \rangle$ is a valid shortcut in the \linf norm,
		whereas in the \lone and \ltwo norm it is not.}
	\label{fig:wedges-and-wave-fronts}
\end{figure}

The following definitions are illustrated in \cref{fig:wedges-and-wave-fronts}.
When starting at $p_i$ and encountering $p_j$ during the traversal,
we denote by $D_{i, j}$ the \emph{local wedge} of $p_i$ and $p_j$
that is the area between the two tangential rays of the unit circle around $p_j$ emanating at $p_i$.
The (global) wedge $W_{i,j}$ is an angular region having its origin at $p_i$.
We define $W_{i, i}$ to be the whole plane and
each $W_{i, j}$ for $j > i$ is essentially
the intersection of all local wedges up to $D_{i, j}$.
We remark that, as mentioned in \cref{sec:guibas},
we apply an extra update step described in \cref{sec:algorithm-outline}
specific to the \fre distance, which may narrow the wedge 
when obtaining $W_{i,j}$ from $W_{i,j-1}$.
Therefore, $W_{i,j} \subseteq \bigcap_{k \in \{i + 1, i + 2, \dots, j\}} D_{i, k}$ holds.
We give a precise inductive definition of the wedge $W_{i,j}$
when we describe the algorithm in \cref{sec:algorithm-outline}.

Let $C_{j}$ be the unit circle around~$p_j$
and let $l_j$ ($r_j$) be the left (right)
tangential point of $C_j$ and $D_{i,j}$.
Between $l_j$ and $r_j$,
there are two arcs of $C_{j}$~-- the \emph{bottom arc} and the \emph{top arc}%
%(assuming that $p_i$ is below~$p_j$)%
\footnote{W.l.o.g., we assume hereunder that $p_i$ is below $p_{i+1}$ and
	therefore at the bottom of a wedge.
	Moreover, w.l.o.g., we assume that $p_{i+1}$ has a distance of at
	least~$\delta$ to~$p_i$ because otherwise, we could ignore all vertices
	following~$p_i$ and having distance~$\le \delta$ to~$p_i$ since they are in
	$\delta$-distance to any shortcut~$\langle p_i, p_j \rangle$.
	Note, though, that a vertex~$p_j$ with $j > i + 1$
	could have distance $\le \delta$ to~$p_i$.
	Then, we define $D_{i, j}$ as the whole plane
	and the whole boundary of $C_j$ as its \emph{top arc}.}.
%\todo{Reviewer 2: bottom/top arcs: This assumes an upwards direction of the curve. Maybe make it more explicit?}
Clearly, any ray emanating at $p_i$ intersects the bottom and the top arc at most once each.
We call the bottom arc of $C_{j}$ between
$l_j$ and $r_j$ the \emph{wave} of~$D_{i,j}$.
We call the region within~$D_{i,j}$ and above and on its wave the \emph{local valid region} of $D_{i,j}$.
We call the region within~$W_{i,j}$ and above and on the wavefront the \emph{valid region}
of $W_{i,j}$ (for $W_{i,i}$ the whole plane).

The wavefront itself is defined inductively.
The \emph{wavefront} of~$W_{i, j}$ (for $j > i$) is the boundary of the intersection of
the valid region of~$W_{i, j - 1}$ and the local valid region of~$D_{i,j}$ within $W_{i, j}$
and excluding the boundary of $W_{i, j}$.
Intuitively, it is the wavefront of $W_{i, j - 1}$ within $W_{i, j}$
where we cut out the bottom arc of~$C_j$.

\begin{figure}
	\centering
	\begin{subfigure}[t]{0.46  \linewidth}
		\centering
		\includegraphics[page=2]{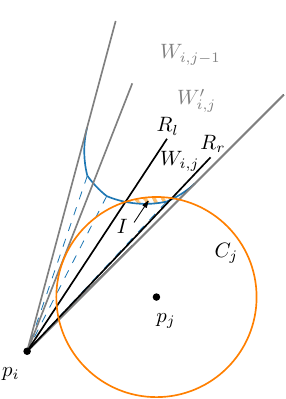}
		\caption{When encountering $p_j$,
			we update the wedge in two steps~--
			even if $p_j$
			lies outside the wedge.}
		\label{fig:update-out}
	\end{subfigure}
	\hfill
	\begin{subfigure}[t]{0.46 \linewidth}
		\centering
		\includegraphics[page=3]{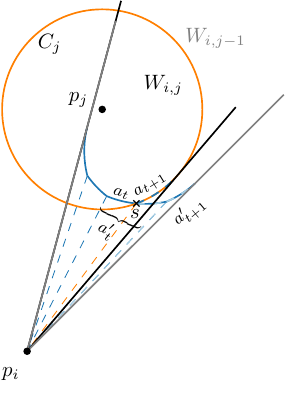}
		\caption{
			Vertex $p_j$ contributes an arc to the new wavefront.
			Here, $\langle p_i, p_j \rangle$ is also a valid shortcut.
		}
		\label{fig:update-in}
	\end{subfigure}
	
	\caption{Updating the wedge and its wavefront in the \ltwo norm.}
	\label{fig:update-wave-front-and-wedge}
\end{figure}

\section{Local-\fre Simplification Algorithm in Near-Quadratic Time}
\label{sec:algorithm}
In this section, we describe how to obtain our
algorithm for polyline simplification under the local \fre distance  running
in near-quadratic time by means of
Melkman and O'Rourke~\cite{melkman1988polygonal} and integrating ideas from Guibas et al.~\cite{Guibas1993}.
%and within the framework of Imai and Iri~\cite{Imai1988}.

\subsection{Outline}
\label{sec:algorithm-outline}
As all Imai--Iri based algorithms, we build the shortcut graph by
traversing the given polyline $n$ times~-- starting once from each vertex~$p_i$
for $i \in \{1, \dots, n\}$
and determining all shortcuts starting at~$p_i$.
For each~$p_i$, we construct a wedge with a wavefront,
%The properties of the wavefront will be discussed in more detail in \cref{sec:wavefront}.
%which are also relevant for the Melkman--O'Rourke algorithm alone.
whose properties are analyzed in more detail in \cref{sec:wavefront}.

Next, we describe how to determine,
for each vertex~$p_i$,
the set of subsequent vertices to which $p_i$ has a valid shortcut.
We traverse the polyline in order $p_{i+1}, p_{i+2}, \dots, p_n$.
During this traversal, we maintain the wedge
in which all valid shortcuts need to lie.
This would, as in the algorithm by Chan and Chin~\cite{Chan1996},
suffice to assure that the \hau distance threshold is not violated
(which is a lower bound for the \fre distance).
To also not exceed the \fre distance threshold, we use the wavefront.
As in the algorithm by Melkman and O'Rourke,
the invariant maintained is that for a valid shortcut from $p_i$ to $p_j$ with $j > i$,
the vertex $p_j$ has to be within the valid region of the wedge~$W_{i, j-1}$.
In this case, we add the directed edge~$p_ip_j$ to the shortcut graph.

Then, regardless of whether $\langle p_i, p_j \rangle$ is a valid shortcut or not,
we first update the wedge~$W_{i, j-1}$
to an intermediate wedge~$W'_{i,j}$ by computing the intersection
between $W_{i, j-1}$ and the local wedge~$D_{i,j}$.
Afterwards, we update the intermediate wedge~$W'_{i, j}$ to the wedge $W_{i, j}$ and
we update the wavefront.
%\footnote{This update step is specific to the \fre distance and is
%	not part of the algorithm by Melkman and O'Rourke.
%	Guibas et al.\ do a similar update step.}.

This update step is illustrated for the \ltwo norm in \cref{fig:update-wave-front-and-wedge}
and for multiple steps and multiple norms in \cref{fig:wedges-and-wave-fronts}.
For the \ltwo norm and for the \lone and \linf norms, we give more detail on this
update step in \cref{sec:wavefront-maintenance,sec:small-wave-fronts-lone-linf}, respectively.
It works as follows.
A valid shortcut $\langle p_i, p_k \rangle$ with $k > j$ in the \fre distance
needs to go through the intersection region~$I$
between the current valid region %(i.e., the wedge behind the wavefront)
and the unit circle~$C_j$ around~$p_j$.
Otherwise, the vertices of the subpolyline from $p_i$ to $p_k$
would be encountered in the wrong order contradicting the
definition of the \fre distance.
Hence, we narrow the intermediate wedge~$W'_{i,j}$ such that
the rays~$R_l$ and $R_r$ emanating at $p_i$ and enclosing~$I$
constitute the wedge~$W_{i,j}$; see \cref{fig:update-out}.
This extra narrowing step is also applied by Guibas et al.\ in
their line stabbing algorithm, but not by Melkman and O'Rourke.
For the \hau distance,
it is irrelevant in which order the intermediate points of
a shortcut are encountered by the shortcut segment.

Thereafter, we update the wavefront as Melkman and O'Rourke do.
The part of the bottom arc of the unit circle~$C_j$ around $p_j$
that is above the current wavefront is included into the new wavefront.
Pictorially, the wavefront is moving upwards.
For an example see \cref{fig:update-in}.
There, we compute the intersection point~$s$ between~$C_j$ and the
wavefront and replace the arcs $a'_t$ and $a'_{t+1}$
of the wavefront by the arcs $a_t$ (which is a part of $a'_t$)
and $a_{t+1}$ (which is a part of $C_j$).
There can be up to two intersection points between~$C_j$ and the wavefront.

If the valid region becomes empty,
we abort the search for further shortcuts from $p_i$.

\subsection{Correctness}
To show that the algorithm works correctly, we prove two things:
that all shortcuts the algorithm finds are valid
(\cref{lem:all-shortcuts-found-are-valid})
and that the algorithm finds all valid shortcuts
(\cref{lem:all-valid-shortcuts-are-found}).
As it is more difficult to show this directly, we first state five helpful lemmas,
which we use throughout the remainder of this manuscript.

\newcounter{bottomArcsDontIntersectTwice}
\setcounter{bottomArcsDontIntersectTwice}{\value{theorem}}
\begin{restatable}[{\hyperref[clm:bottom-arcs-dont-intersect-twice*]{\appmark}}]{lemma}{bottomArcsDontIntersectTwice}
	\label{clm:bottom-arcs-dont-intersect-twice}
	Given two unit circles and a point~$p$ outside of the unit circles.
	If the two bottom arcs (with respect to~$p$)
	intersect, then the second intersection point 
	is between their top arcs.
\end{restatable}

\newcounter{twoIntersectionFromBelow}
\setcounter{twoIntersectionFromBelow}{\value{theorem}}
\begin{restatable}[{\hyperref[lem:unit-circle-wave-front-two-intersections-from-below*]{\appmark}}]{lemma}{twoIntersectionFromBelow}
	\label{lem:unit-circle-wave-front-two-intersections-from-below}
	If a unit circle $C$ intersects the wavefront more than once,
	then on the left side of the leftmost intersection point~$s_1$
	(relative to rays originating in $p_i$)
	and on the right side of the rightmost intersection point~$s_2$,
	$C$ is below the wavefront.
	In other words, the intersection pattern depicted in
	\cref{fig:unit-circle-intersects-wave-front-not-from-above} cannot occur.
\end{restatable}

\newcounter{waveFrontInsideUnitCircle}
\setcounter{waveFrontInsideUnitCircle}{\value{theorem}}
\begin{restatable}[{\hyperref[lem:unit-circles-of-the-wave-front-contain-whole-wave-front*]{\appmark}}]{lemma}{waveFrontInsideUnitCircle}
	\label{lem:unit-circles-of-the-wave-front-contain-whole-wave-front}
	Consider the wavefront of~$W_{i,k}$.
	For every vertex $p_j$ ($i < j \le k$) whose unit circle~$C_j$ contributes
	an arc of the wavefront of~$W_{i,k}$, the wavefront of~$W_{i,k}$
	lies completely inside~$C_j$.
\end{restatable}
\cref{lem:unit-circles-of-the-wave-front-contain-whole-wave-front} directly implies the following lemma.
\begin{lemma}
	\label{lem:point-on-wave-front-reachable-by-pj}
	Let $q$ be a point lying on the wavefront of the wedge~$W_{i, j}$.
	Then, $d(p_j, q) \le \delta$.
\end{lemma}
%
%We use one more useful property before we show that exactly
%the valid shortcuts are found by our algorithm.

\begin{lemma}
	\label{lem:wave-front-only-moves-away}
	Let $R$ be a ray emanating at $p_i$ and lying inside the
	wedges $W_{i, j}$ and $W_{i, k}$ for some $i < j < k$.
	Moreover, let $q_j$ and $q_k$ be the intersection points
	between $R$ and the wavefronts of $W_{i, j}$ and $W_{i, k}$, respectively.
	Then, $d(p_i, q_j) \le d(p_i, q_k)$.
\end{lemma}

\begin{proof}
	Assume for contradiction that $d(p_i, q_j) > d(p_i, q_k)$.
	Then, $q_k$ is below the wavefront of $W_{i,j}$
	and, hence, $q_k$ does not lie in the valid region of $W_{i,j}$
	but in the valid region of~$W_{i,k}$.
	However, the valid region of~$W_{i,k}$ is the intersection
	of the local valid region of~$D_{i,k}$ and all previous valid regions
	including $W_{i,j}$ and, thus, the valid region of $W_{i,k}$ is a subset of $W_{i,j}$.
\end{proof}

Putting \cref{lem:wave-front-only-moves-away} in other words,
the wavefront may only move away but never towards $p_i$
during the execution of the algorithm.
We are now ready to prove the correctness of the algorithm
by the following two lemmas.

\begin{lemma}
	\label{lem:all-shortcuts-found-are-valid}
	Any shortcut found by the algorithm
	is valid under the local \fre distance
	and any \lp norm with $p \in [1, \infty]$.
\end{lemma}

\begin{proof}
	Let $\langle p_i, p_k \rangle$ be a shortcut found by the algorithm.
	We show that there is a mapping of the vertices $\langle p_{i+1}, p_{i+2}, \dots, p_{k-1} \rangle$
	onto points $\langle m_{i+1}, m_{i+2}, \dots, m_{k-1} \rangle$,
	such that $m_j \in \overline{p_i p_k}$ and $d(p_j, m_j) \le \delta$
	for every $j \in \{i+1, \dots, k-1\}$, and
	$m_{j}$ precedes or equals $m_{j + 1}$ for every $j \in \{i+1, \dots, k-2\}$
	when traversing $\overline{p_i p_k}$ from $p_i$ to $p_k$.
	Clearly, this implies that also the \fre distance between each pair of
	line segments $\overline{p_{j} p_{j+1}}$ and $\overline{m_j m_{j+1}}$
	is at most $\delta$ and, hence, $\langle p_i, p_k \rangle$ is a valid shortcut.
	In the remainder of this proof, we describe how to obtain
	$m_{i+1}, m_{i+2}, \dots, m_{k-1} \in \overline{p_i p_k}$.
	To this end, we consider the wedge $W_{i, j}$
	and the corresponding wavefront for each $j \in \{i+1, \dots, k-1\}$, i.e.,
	for each intermediate step when executing the algorithm.
	By construction of the algorithm, $\overline{p_i p_k}$ lies inside the wedge $W_{i, j}$
	and $p_k$ lies above its wavefront
	(since $p_k$ lies in the valid region of $W_{i, k - 1}$
	and, by \cref{lem:wave-front-only-moves-away},
	the wavefront has never moved towards~$p_i$).
	Let $m_j$ be the intersection point of $\overline{p_i p_k}$ and
	the wavefront of $W_{i, j}$.
	By \cref{lem:point-on-wave-front-reachable-by-pj}, $d(p_j, m_j) \le \delta$.
	Moreover, by \cref{lem:wave-front-only-moves-away},
	$m_{j}$ precedes or equals $m_{j + 1}$ for any $j \in \{i+1, \dots, k-2\}$
	when traversing $\overline{p_i p_k}$ from $p_i$ to~$p_k$.
\end{proof}

\begin{lemma}
	\label{lem:all-valid-shortcuts-are-found}
	All valid shortcuts under the local \fre distance and any \lp norm with
	$p \in [1, \infty]$ are found by the algorithm.
\end{lemma}

\begin{proof}
	Suppose for the sake of a contradiction that there is a valid shortcut~$\langle p_i, p_k \rangle$
	that was not found by the algorithm.
	
	If $p_k$ lay outside of $\bigcap_{j \in \{i + 1, i + 2, \dots, k-1\}} D_{i, j}$,
	then there would be some $p_{j'}$ with $i < j' < k$ such that
	$d(p_{j'}, \overline{p_i p_k}) > \delta$.
	So, as in the algorithm by Chan and Chin~\cite{Chan1996},
	already the \hau distance requirement would be violated
	and $\langle p_i, p_k \rangle$ would be no valid shortcut.
	Hence, $p_k$ lies inside~$\bigcap_{j \in \{i + 1, i + 2, \dots, k-1\}} D_{i, j}$.
	
	Suppose now that $p_k$ lies inside $\bigcap_{j \in \{i + 1, i + 2, \dots, k-1\}} D_{i, j}$ but outside $W_{i, k - 1}$.
	W.l.o.g.\ $p_k$ lies to the left of the wedge $W_{i, k - 1}$.
	We know that there is some $p_j$ with $i < j < k$ for which
	the extra narrowing step from \cref{sec:algorithm-outline} has been applied
	such that $p_k$ lies to the left of $W_{i, j}$.
	For constructing $W_{i, j}$, we have considered the intersection area~$I$
	between $C_j$ and the wavefront of $W_{i, j-1}$.
	The left endpoint of~$I$ lies on the boundary of $W_{i, j}$
	and is the intersection point between $C_j$ and an arc of the wavefront
	of $W_{i, j-1}$ belonging to a vertex $p_{j'}$ with $i < j' < j$.
	Now consider the ray $R$ that we obtain by extending $\overline{p_i p_k}$ at $p_k$.
	When traversing~$R$, we first enter and leave the interior of $C_{j}$
	before we enter the interior of~$C_{j'}$.
	Hence, the \fre distance between $\overline{p_i p_k}$ and $L[p_i, p_k]$
	is greater than $\delta$ due to the order of $p_{j'}$ and $p_j$ within $L[p_i, p_k]$.
	Therefore, $p_k$ lies inside~$W_{i, k-1}$.
	
	Finally, suppose that $p_k$ lies inside $W_{i, k-1}$ but not in
	the valid region, i.e., $p_k$ lies below the wavefront of $W_{i, k - 1}$.
	%	(we assume w.l.o.g.\ that $p_i$ is the bottommost point of $W_{i, k - 1}$).
	Since $p_k$ is below the wavefront, the line segment~$\overline{p_i p_k}$
	does not intersect the wavefront (otherwise, we would violate
	\cref{lem:wave-front-has-at-most-one-intersection-point-with-a-ray};
	see below).
	Again, consider the ray $R$ that we obtain by extending $\overline{p_i p_k}$ at $p_k$.
	Let the intersection point of $R$ and the wavefront of $W_{i, k - 1}$ be~$w$.
	The point~$w$ lies on an arc of the wavefront.
	This arc is part of the bottom arc of a unit circle~$C_j$
	belonging to some $p_j$ with $i < j < k$.
	Since it is the bottom arc, $p_k$ lies outside $C_j$ and $d(p_k, p_j) > \delta$.
	
	Therefore, $p_k$ lies in the valid region of $W_{i, k - 1}$.
	However, these are precisely the vertices for which the algorithm adds a shortcut.
\end{proof}

\section{The Wavefront Data Structure}
\label{sec:wavefront}
At the heart of the algorithm lies the maintenance of the wavefront.
To show that the algorithm can be implemented to run in $\oh(n^2 \log n)$ time,
we next analyze the properties of the wavefront and discuss how to store
and update it using a suitable (simple) data~structure.

\subsection{Size of the Wavefront}
\label{sec:wavefront-size}
We  first prove that the wavefront always has a size in $\oh(n)$. This insight is based on two properties proven in the following lemmas.
%For ease of exposition, we assume that we never compare two identical unit circles and also that our polylines consist of $n$ distinct points.
\begin{lemma}
	\label{lem:wave-front-has-at-most-one-intersection-point-with-a-ray}
	Any ray emanating at $p_i$ intersects the wavefront at most once.
\end{lemma}
\begin{proof}
	We prove this statement inductively.
	As $W_{i, i + 1} = D_{i, i + 1}$, consider the wave of~$D_{i, i + 1}$.
	Since the unit circle in any \lp norm for $p \in [1, \infty]$ is convex,
	any ray emanating at $p_i$ intersects a unit circle at most twice.
	The first intersection is with the bottom
	arc of the unit circle~$C_{i+1}$ and the second intersection
	is with the top arc of~$C_{i+1}$.
	As the wave of $D_{i, i + 1}$ is defined as the bottom arc of~$C_{i+1}$,
	any ray emanating at~$p_i$ intersects the wave of $D_{i, i + 1}$ at most once.
	
	It remains to show the induction step for all $j > i + 1$.
	By the induction hypothesis, we know that any ray emanating at $p_i$
	intersects the wavefront of $W_{i, j - 1}$ at most once.
	The wavefront of~$W_{i, j}$ is the boundary of the intersection of
	the valid region of~$W_{i, j - 1}$ and the local valid region of~$D_{i,j}$.
	Consider a ray~$R$ originating at $p_i$.
	The ray~$R$ enters the valid region of~$W_{i, j - 1}$ at most
	at one point~$q$ where it also intersects the wavefront
	of~$W_{i, j - 1}$, and it enters the local valid region of~$D_{i,j}$
	at most at one point~$q'$ where it also intersects the wave of~$D_{i,j}$.
	Hence, $R$ enters the intersection of the valid region of~$W_{i, j - 1}$ and the local valid
	region of~$D_{i,j}$ at most at one point~-- namely either at $q$ or at~$q'$ (or $q = q'$).
	This is the only point of the wavefront of~$W_{i, j}$ that is shared with~$R$.
\end{proof}

We can make a similar statement for unit circles.
The number of intersection points between a unit circle and the wavefront
is important for updating the wavefront.
%Intuitively, the following statement holds because
%when we traverse the wavefront from left to right...
\todo{as noted by a reviewer: maybe add an intuitive argument
why the following lemma is true.}

\newcounter{atMostTwice}
\setcounter{atMostTwice}{\value{theorem}}
\begin{restatable}[{\hyperref[lem:wave-front-has-at-most-two-intersection-points-with-another-circle*]{\appmark}}]{lemma}{atMostTwice}
	\label{lem:wave-front-has-at-most-two-intersection-points-with-another-circle}
	Any unit circle of radius~$\delta$ intersects the wavefront at most twice.
\end{restatable}

From \cref{lem:wave-front-has-at-most-two-intersection-points-with-another-circle}
it follows that in each step, the size of the wavefront increases at most by~2.
This leads us to the following lemma.

\begin{lemma}
	\label{lem:wave-front-has-at-most-linear-size}
	The wavefront consists of at most~$\oh(n)$ arcs under any L$_{p \in (1, \infty)}$ norm.
\end{lemma}
\begin{proof}
	According to the inductive definition, we start with a wavefront consisting of one arc.
	Now in each step where we extend the wavefront,
	we consider the intersection between the current valid region and a local valid region~--
	one	is defined by the current wavefront, the other is defined by a single arc~$a$.
	This is the intersection between the current wavefront
	and the unit circle on which~$a$ lies.
	By \cref{lem:wave-front-has-at-most-two-intersection-points-with-another-circle},
	we know that there are at most two intersection points.
	This means, the number of arcs on the wavefront increases by at most two.
	In the worst case, we start at vertex~$p_1$ and adjust the wavefront
	$n-1$ times until we have created the wavefront of~$W_{1, n}$.
	Therefore, any wavefront consists of at most $2 n - 3 \in \oh(n)$ arcs\footnote{%
		One can further observe that,
		by \cref{lem:unit-circle-wave-front-two-intersections-from-below},
		the number of arcs on the wavefront increases actually
		by at most one per vertex~$p_j$ ($j \in \{2, \dots, n\}$).
		This means any wavefront consists of at most $n-1$ arcs.
	}.
\end{proof}

\subsection{Wavefront Maintenance under the \ltwo Norm (Euclidean Norm)}
\label{sec:wavefront-maintenance}
As there might be a linear number of  arcs on the
wavefront,  we cannot simply iterate over all arcs
in each step of the algorithm since this would  require cubic time in total.
Therefore, we employ a data structure that allows for
querying, inserting, and removing an object in logarithmic time.
Similar to Melkman and O'Rourke, we use a
balanced search tree (e.g., a red-black tree)
where we store the circular arcs~\footnote{%
	To represent a circular arc, we store its corresponding unit circle center and the points where the arc starts and ends.}
of the wavefront.
The keys according to which the circular arcs are arranged in the search tree
are the angles of their starting points with respect to~$p_i$.
These angles cover a range of less than~$\pi$, hence, we may rotate the drawing
when computing the wavefronts of a vertex~$p_i$ to avoid ``jumps'' from $2 \pi$ to $0$.
We can then locate a point~$p_j$ relative to the wavefront
and add or remove an arc on the wavefront in logarithmic time.
%More precisely, consider an arc~$a$ of the wavefront.
%The line segment between the (left) starting point of $a$ and $p_i$
%has some angle~$\alpha$ at $p_i$. We use this $\alpha$ as key for~$a$.

Note that, different from Melkman and O'Rourke and similar to Guibas et al.,
we have an additional update step
where we determine the intersection region~$I$ and
potentially make the wedge narrower.
We show that we can update the wavefront in amortized logarithmic time
using a simple case distinction.
We compute the intersection area~$I$ only implicitly.
For an overview, see \cref{fig:update-wave-front}.

\begin{figure}
	\centering
	\begin{subfigure}[t]{.31 \linewidth}
		\centering
		\includegraphics[page=4, trim = 0 85 0 0, clip]{update-wave-front}
		\smallskip
		\caption{\centering Case \textsf{TB}. \\ (cannot occur)}
		\label{fig:update-wave-front-case-tb}
	\end{subfigure}
	\hfill
	\begin{subfigure}[t]{.31\linewidth}
		\centering
		\includegraphics[page=7, trim = 0 85 0 0, clip]{update-wave-front}
		\smallskip
		\caption{\centering Case \textsf{TM}.}
		\label{fig:update-wave-front-case-tm}
	\end{subfigure}
	\hfill
	\begin{subfigure}[t]{.31\linewidth}
		\centering
		\includegraphics[page=9, trim = 0 85 0 0, clip]{update-wave-front}
		\smallskip
		\caption{\centering Case \textsf{TT}.}
		\label{fig:update-wave-front-case-tt}
	\end{subfigure}
	
	\smallskip
	
	\begin{subfigure}[t]{.31 \linewidth}
		\centering
		\includegraphics[page=2, trim = 0 85 0 0, clip]{update-wave-front}
		\smallskip
		\caption{\centering Case \textsf{MB}.}
		\label{fig:update-wave-front-case-mb}
	\end{subfigure}
	\hfill
	\begin{subfigure}[t]{.31\linewidth}
		\centering
		\includegraphics[page=6, trim = 0 85 0 0, clip]{update-wave-front}
		\smallskip
		\caption{\centering Case \textsf{MM}.}
		\label{fig:update-wave-front-case-mm}
	\end{subfigure}
	\hfill
	\begin{subfigure}[t]{.31\linewidth}
		\centering
		\includegraphics[page=8, trim = 0 85 0 0, clip]{update-wave-front}
		\smallskip
		\caption{\centering Case \textsf{MT}.}
		\label{fig:update-wave-front-case-mt}
	\end{subfigure}
	
	\bigskip
	
	\begin{subfigure}[t]{.31 \linewidth}
		\centering
		\includegraphics[page=1, trim = 0 85 0 0, clip]{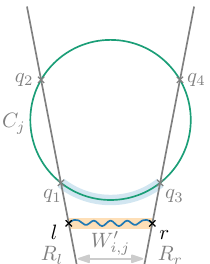}
		\smallskip
		\caption{\centering Case \textsf{BB}.}
		\label{fig:update-wave-front-case-bb}
	\end{subfigure}
	\hfill
	\begin{subfigure}[t]{.31\linewidth}
		\centering
		\includegraphics[page=3, trim = 0 85 0 0, clip]{update-wave-front}
		\smallskip
		\caption{\centering Case \textsf{BM}.}
		\label{fig:update-wave-front-case-bm}
	\end{subfigure}
	\hfill
	\begin{subfigure}[t]{.31\linewidth}
		\centering
		\includegraphics[page=5, trim = 0 85 0 0, clip]{update-wave-front}
		\smallskip
		\caption{\centering Case \textsf{BT}. \\ (cannot occur)}
		\label{fig:update-wave-front-case-bt}
	\end{subfigure}
	
	\smallskip
	
	\caption{
		Cases for updating the wavefront when a vertex $p_j$, whose 
		unit circle is $C_j$, is added.
		The wavefront of the previous step is illustrated as a blue wavy line.
		Green background color (everything inside $C_j$) highlights the parts that remain part of the wavefront
		and red and orange background color (everything outside $C_j$) highlight
		the parts that are removed from the wavefront.
		The part of the bottom arc of $C_j$ that becomes part of the wavefront is
		highlighted by blue background color.
	}
	\label{fig:update-wave-front}
\end{figure}

Say we are computing all valid shortcuts starting at $p_i$ and we are
currently processing a vertex $p_j$, which is the center of the unit circle~$C_j$.
We have already constructed the intermediate wedge $W'_{i, j}$
and clipped the wavefront of $W_{i, j-1}$ along the left and the
right bounding rays $R_l$ and $R_r$ of $W'_{i, j}$.
For this clipping, we may have removed a linear number of arcs,
however, over all iterations we remove every arc at most once.
Now, both $R_l$ and $R_r$ intersect~$C_j$ twice or touch~$C_j$.
Let $q_1$ and $q_2$ denote the intersection points
between $R_l$ and $C_j$ (where $q_1$ is on the bottom arc of~$C_j$).
Similarly, let $q_3$ and $q_4$ denote the intersection points
between $R_r$ and $C_j$ (where $q_3$ is on the bottom arc of~$C_j$).
Moreover, let $l$ and $r$ denote the intersection point between the wavefront
and $R_l$ and between the wavefront and $R_r$, respectively.

The relative positions of $l$, $q_1$, and $q_2$ on $R_l$
and the relative positions of $r$, $q_3$, and $q_4$ on $R_r$
determine where the intersection points~$s_1$ and $s_2$ (if they exist)
between $C_j$ and the wavefront of $W_{i, j-1}$ can lie.
(Recall that there are at most two intersection points by
\cref{lem:wave-front-has-at-most-two-intersection-points-with-another-circle}.)
In the following, we write $a \prec b$ if $a$ is below $b$ along the ray $R_l$ or $R_r$.
If $q_1 = l$ or $q_2 = l$, then we proceed as if $R_l$ was moved
to the right by a tiny bit (symmetrically as if $R_r$ was moved to the left).
Thus, at such a point, the angle of the incident arc of the wavefront
and~$C_j$ matters for the order.
For the degenerate case $q_1 = l = q_2$, which includes a touching point
between $R_l$ and $C_j$, we hence assume $q_1 \prec l \prec q_2$.

Next, we consider all orderings of $l$, $r$, $q_1$, $q_2$, $q_3$, and $q_4$.
This gives rise to the following nine cases.
We remark that two of them (Case~\textsf{TB} and Case~\textsf{BT})
cannot occur and two pairs of the remaining cases are symmetric,
which leaves essentially five different cases.

\medskip\noindent\textbf{(Case~\textsf{TB}:)} $q_1 \prec q_2 \prec l$ and
$r \prec q_3 \prec q_4$; see \cref{fig:update-wave-front-case-tb}.
This case cannot occur.
Suppose for a contradiction that we have this configuration.
Then, there are precisely two intersection points~$s_1$ and $s_2$ between the wavefront and $C_j$.
The left intersection point~$s_1$ is between the top arc of~$C_j$
and an arc $a_k$ of the wavefront, which is part of the bottom arc of a unit circle~$C_k$
belonging to a vertex $p_k$ with $i < k < j$.
By \cref{lem:unit-circles-of-the-wave-front-contain-whole-wave-front},
$C_k$ contains the whole wavefront, thus including~$s_2$,
which means $C_k$ and $C_j$ intersect a
second time such that this intersection point is to the left of $s_2$.
Then, however, $C_k$ and $C_j$ intersect once with both bottom arcs
and once with a bottom and a top arc~-- a contradiction to~\cref{clm:bottom-arcs-dont-intersect-twice}.
For more details on this argument, see in \cref{app:algorithm}
the proof of \cref{lem:unit-circles-of-the-wave-front-contain-whole-wave-front}
and \cref{fig:wave-front-inside-unit-circle}.

\medskip\noindent\textbf{Case~\textsf{TM}:} $q_1 \prec q_2 \prec l$ and
$q_3 \prec r \prec q_4$; see \cref{fig:update-wave-front-case-tm}.
There is an intersection point~$s_1$ between $C_j$ and the wavefront.
We traverse\footnote{%
	In the following we just say for short, ``we traverse the wavefront starting at $l$''.}
the arcs in the balanced search tree representing the wavefront starting at
the leftmost arc, which in turn starts at point $l$,
and remove all arcs that we encounter until we find the intersection
point~$s_1$ between $C_j$ and an arc~$a$ of the wavefront.
We update $a$ to start at $s_1$ and the left bounding ray of $W_{i, j}$ to go through $s_1$.
There cannot be a second intersection point because otherwise
there would also be a third intersection point between a unit circle and the wavefront.

\medskip\noindent\textbf{Case~\textsf{TT}:} $q_1 \prec q_2 \prec l$ and
$q_3 \prec q_4 \prec r$; see \cref{fig:update-wave-front-case-tt}.
In this case, we either have two or no intersection points between the wavefront and~$C_j$.
We traverse the wavefront starting at $l$ and remove all arcs that
we encounter and that do not intersect~$C_j$.
If we do not find any intersection point but reach~$r$, then
there cannot be any further valid shortcut starting at $v_i$ and we abort.
Otherwise, we have found $s_1$ and proceed symmetrically at $r$ to find $s_2$.

\medskip\noindent\textbf{Case~\textsf{MB}:} $q_1 \prec l \prec q_2$ and
$r \prec q_3 \prec q_4$; see \cref{fig:update-wave-front-case-mb}.
There is precisely one intersection point $s_1$ between the wavefront and the bottom arc of $C_j$.
We traverse the wavefront starting at $r$ and remove all arcs that
we encounter and that do not intersect~$C_j$ until we have found~$s_1$.
We clip the arc of the wavefront at $s_1$ and append the bottom arc
of~$C_j$ between $s_1$ and $r$ to the wavefront.

\medskip\noindent\textbf{Case~\textsf{MM}:} $q_1 \prec l \prec q_2$ and
$q_3 \prec r \prec q_4$; see \cref{fig:update-wave-front-case-mm}.
There are either two or no intersection points between the wavefront and $C_j$.
If there are two intersection points,
then they are on the bottom arc of $C_j$ as otherwise,
it would contradict \cref{lem:unit-circle-wave-front-two-intersections-from-below}.
The order of the arcs on the wavefront around $p_i$ 
is reverse to the order of the corresponding unit circle centers around $p_i$~\cite{melkman1988polygonal}.
By binary search, we determine the arcs $a_k$ and $a_{k+1}$
such that $p_j$ lies in between the corresponding unit circle centers
of $a_k$ and $a_{k+1}$ with respect to the angle around $p_i$.
If $a_k$ and $a_{k+1}$ are completely contained inside $C_j$,
then there is no intersection point and the wavefront remains unchanged.

Otherwise, we traverse the wavefront starting at $a_k$ to the left
until we have found an arc intersecting $C_j$, which gives us $s_1$.
We remove all arcs along the way and split the arc containing $s_1$ at $s_1$.
Symmetrically, we traverse the wavefront starting at $a_{k+1}$ to the right
to find $s_2$.
Finally, at the resulting gap, we insert the arc of $C_j$
between $s_1$ and $s_2$ into the wavefront.

\medskip\noindent\textbf{Case~\textsf{MT}:} $q_1 \prec l \prec q_2$ and
$q_3 \prec q_4 \prec r$; see \cref{fig:update-wave-front-case-mt}.
This case is symmetric to Case~\textsf{TM}.

\medskip\noindent\textbf{Case~\textsf{BB}:} $l \prec q_1 \prec q_2$ and
$r \prec q_3 \prec q_4$; see \cref{fig:update-wave-front-case-bb}.
There is no intersection point between the wavefront and $C_j$.
There cannot be a single intersection point and if there were two
intersection points, it would contradict
\cref{lem:unit-circle-wave-front-two-intersections-from-below}.
We replace the whole wavefront by the arc of $C_j$ from $q_1$ to $q_3$.

\medskip\noindent\textbf{Case~\textsf{BM}:} $l \prec q_1 \prec q_2$ and
$q_3 \prec r \prec q_4$; see \cref{fig:update-wave-front-case-bm}.
This case is symmetric to Case~\textsf{MB}.

\medskip\noindent\textbf{(Case~\textsf{BT}:)} $l \prec q_1 \prec q_2$ and
$q_3 \prec q_4 \prec r$; see \cref{fig:update-wave-front-case-bt}.
Since this configuration is symmetric to Case~\textsf{TB}, this case also cannot occur.

\medskip

Note that we only do binary search in logarithmic time
or if we traverse multiple arcs, we remove them from the wavefront.
During the whole process, we add,
for any vertex $p_j$ with $j > i$, at most one arc to the wavefront.
Therefore, we conclude \cref{lem:update-wave-front-amortized-log}.

\begin{lemma}
	\label{lem:update-wave-front-amortized-log}
	Given a two-dimensional $n$-vertex polyline $L$ and a vertex $p \in L$,
	we can find all valid shortcuts under the local \fre distance
	starting at $p$ in $\oh(n \log n)$ time.
\end{lemma}

This update process in amortized logarithmic time per vertex
is tailored specifically for this polyline simplification algorithm.
We remark, however, that the more general problem of determining the two, one, or zero
intersection points between a unit circle and the wavefront can also be
accomplished in logarithmic time by a recursive case distinction.

%%% remark is commented out, see the end of the tex file for more information
\begin{comment}
We do not use it here (although we could alternatively update the wavefront using this procedure),
but for completeness and~-- as it might be relevant for using
the wavefront in a different context~-- we have added it to the appendix.
It is a more involved case distinction used while traversing
a balanced binary search tree storing the wavefront.

\newcounter{logTimeIntersection}
\setcounter{logTimeIntersection}{\value{theorem}}
\begin{restatable}[{\hyperref[lem:log-time-intersection-unit-circle-wave-front*]{\appmark}}]{remark}{logTimeIntersection}
	\label{lem:log-time-intersection-unit-circle-wave-front}
	Given a unit circle~$C$,
	we can determine the up-to-two intersection points between~$C$
	and a wavefront of size $n$ in $\oh(\log n)$ time in the \ltwo norm.
\end{restatable}
\end{comment}

Now we have all ingredients to prove our main theorem.

\ltwotheorem*

\begin{proof}
	According to \cref{lem:all-shortcuts-found-are-valid,lem:all-valid-shortcuts-are-found},
	the algorithm we describe in \cref{sec:algorithm-outline}
	finds all valid shortcuts.
	It remains to analyze the runtime. 
	We consider each of the $n$ vertices as potential shortcut starting point~$p_i$.	
	When we encounter a vertex~$p_j$ with $j > i$, we determine
	in logarithmic time whether it is in the valid region.
	We do this by computing the ray emanating at $p_i$ and going through~$p_j$,
	and querying the arc it intersects in the wavefront.
	Then, using the case distinction of \cref{lem:update-wave-front-amortized-log},
	we update the wavefront and the wedge.
	This needs amortized logarithmic time and over all steps $\oh(n \log n)$ time.
	
%	Overall, we hence need $\oh(n \log n)$ time per vertex.
	Consequently, we construct the shortcut graph in $\oh(n^2 \log n)$ time.
	In the resulting shortcut graph, we can find an optimal polyline simplification
	by finding a shortest path in $\oh(n^2)$ time.
	
	Regarding space consumption, we observe that the wavefront maintenance only requires linear space at any time.
	As we can compute the set of outgoing shortcuts of each vertex $p_i$ individually,
	we can also easily apply the space reduction approach described for
	the Imai--Iri algorithm in \cref{sec:imai-iri}
	to get an overall space consumption in $\oh(n)$.
\end{proof}

\subsection{Extension to General \lp Norms}\label{sec:allp}
We can use   our data structure for the \ltwo norm also for the \lp norm for $p \in (1, \infty)$.
However, we should take this with a grain of salt as computing the intersection points
between lines and unit circles and between pairs of unit circles in the \lp norm
for $p \in (1, \infty) \setminus \{2\}$ may involve solving equations of degree~$p$,
which may raise numerical questions for the required precision.
To avoid this, we assume for the following corollary that we can determine
intersection points between unit circles and lines (or another unit circle)
in all \lp norms in constant time.

\begin{corollary}
	\label{clm:lp-n2logn}
	A two-dimensional $n$-vertex polyline can be simplified optimally
	under the local \fre distance in the \lp norm for $p \in (1, \infty)$
	in $\oh(n^2 \log n)$ time and $\oh(n)$ space,
	given that we can compute intersection points between a unit circle
	and a line and between two unit circles in constant time.
\end{corollary}

\section{Use Cases with Small Wavefronts}
\label{sec:small-wave-fronts}
The most complicated part of the algorithm is the efficient maintenance of the wavefront.
But in case the wavefront has low complexity,
we do not need any dynamic binary search data structure to make updates,
but we can simply iterate over a linked list representing the whole wavefront
to determine the relevant intersection points and to perform the dynamic changes.
We next discuss use cases where the wavefront complexity is provably small.

\subsection{\lone and \linf Norm (Manhattan and Maximum Norm)}
\label{sec:small-wave-fronts-lone-linf}
In the \lone and \linf norm, the unit circles are actually square-shaped. Thus, the wavefront consists of a sequence of orthogonal line segments.
As we show in the next lemma, this reduces the potential size of the wavefront significantly.
\begin{lemma}
	\label{lem:wave-front-l1-linf-two-segments}
	In the \lone norm and the \linf norm, the wavefront always consists of either one or two
	(orthogonal) straight line segments.
	These straight line segments are horizontal or vertical in the \linf norm
	and rotated by 45 degrees in the \lone norm.
\end{lemma}
\begin{proof}
	We show this claim inductively.
	For $W_{i, i + 1} = D_{i, i + 1}$ it is just the
	bottom arc of a square (the unit circle in \lone or \linf).
	Clearly, this is either one or two orthogonal line segments~--
	horizontal or vertical line segments in the \linf norm
	and line segments rotated by 45 degrees in the \lone norm.
	
	When we compute the wavefront of $W_{i,j}$,
	we compute the intersection of the valid region of $W_{i, j-1}$
	(which is bounded by one or two orthogonal line segments by the induction hypothesis)
	and the local valid region of $D_{i,j}$ (which is bounded by
	one or two line segments parallel to the ones of $W_{i, j-1}$).
	Computing the boundary of this intersection in the \linf norm
	can be done by computing the intersection of two axis-parallel rectangles.
	The intersection of two axis-parallel rectangles is again an axis-parallel rectangle.
	In the \lone norm, the situation is the same but rotated by 45~degrees.
\end{proof}
%Let us describe how to check containment in the valid region and how we can update
%the wedge and the wavefront.
%Say we are computing the shortcuts originating at $p_i$ and we have reached $p_j$.
%We first check in constant time if $p_j$ is in the wedge $W_{i, j-1}$
%by comparing the radial angle of $p_j$ relative to $p_i$ with the radial angles
%of the boundaries of the wedge~$W_{i, j-1}$.
%Then, we check in constant time that $p_j$ is beyond both of the underlying
%lines of the up-to-two line segments of the wavefront.
%
%We update the wavefront by the following operations.
%The unit circle around $p_j$ is bounded by four line segments that are horizontal or vertical
%(in the \linf norm) or rotated by 45 degrees (in the \lone norm).
%Therefore, they are parallel to the up-to-two line segments of the wavefront
%of $W_{i, j-1}$ (see \cref{lem:wave-front-l1-linf-two-segments}) and, therefore,
%we can update the wedge and the wavefront in constant time just
%by computing the intersection points of these up to six line segments.
We hence obtain the following theorem.

\lonelinftheorem*

\subsection{Light Polylines}
In \cref{lem:unit-circles-of-the-wave-front-contain-whole-wave-front},
we have observed that for any vertex $p$ whose unit circle $C$ contributes
to the current wavefront, $C$ actually contains the complete wavefront.
%This also implies that two vertices with their unit circles both contributing to the current wavefront have to be within a distance of $2\delta$ of each other.
Thus, if two vertices contribute an arc to the current wavefront,
they are within a distance of $2\delta$, i.e., inside a unit circle.
To end up with a complex wavefront, there hence need to be many vertices in close proximity
(and they also need to occur in a specific pattern for all
of their unit circles contributing to the wavefront simultaneously).
Accordingly, if we consider polylines with bounded vertex density,
the wavefront complexity is bounded as well.
To formalize this, we introduce the natural class of
%polylines which we call
\emph{$\nu$-light} polylines.
\begin{definition}
	A polyline $L$ in $d$ dimensions is \emph{$\nu$-light}
	if for any $k \in \{2, \dots, n\}$,
	no set (in particular not the closest set) of $k$ vertices of~$L$
	lies in a ball of radius less than $(k / \nu)^{1/d}$.
\end{definition}

Before we exploit the properties of $\nu$-lightness in the context of polyline
simplification, we want to gain some more intuition behind this definition.
If a polyline in two dimensions is $\nu$-light,
this guarantees that the vertices are somewhat well distributed:
the closest pair of vertices has a distance of at least $2 \cdot \sqrt{2/\nu}$,
the closest triplet of vertices has a surrounding circle of diameter
at least $2 \cdot \sqrt{3/\nu}$ and so on.
An alternative (and maybe more intuitive) definition is that
a polyline $L$ is $\nu$-light if for any point $p \in \mathbb{R}^d$
and any radius $r > 0$, the number of vertices of $L$
inside the ball $B_r(p)$ of radius $r$ centered at $p$ is at most $\max \{\nu r^d, 1\}$.
This shows the connection to the related concepts of $c$-packed curves,
$\phi$-low density curves, and $\kappa$-bounded curves
studied in previous work to show improved bounds,
e.g., for computing the (approximate) \fre distance between two curves~\cite{driemel2012approximating}.
The main difference is that these classifications do not distinguish between polyline vertices
and points on the straight line segments in between,
which, however, is important in our scenario.

Now let us revisit our algorithm for polyline simplification with
distance threshold $\delta$ for the local \fre distance.
In a ball of radius  $r=\delta$ (i.e., a unit circle),
a two-dimensional $\nu$-light polyline has $\oh(\nu \delta^2)$ vertices.
Hence, for any constant choice of~$\delta$,
the wavefront complexity is $\oh(\nu)$.
The running time of the algorithm is then in $\oh(n^2 \log \nu)$,
or in $\oh(n^2 \nu)$ when omitting the tree data structure.
So if $\nu \in \oh(1)$, the resulting running time is quadratic even without using a dedicated dynamic data structure.

\begin{theorem}
	A two-dimensional $\nu$-light $n$-vertex polyline can be simplified optimally under the local \fre distance
	in the \lp norm for $p \in [1, \infty]$
	in $\oh(n^2)$ time and $\oh(n)$ space,
	given that we can compute intersection points between a unit circle
	and a line and between two unit circles in constant time.
\end{theorem}

%In practical applications, we generally expect $\nu \in \oh(1)$
%and, hence, a quadratic running time.
%\todo{Ich habe den Satz mal angehaengt. Was meinst du?
%So wirklich intuitiv spuert man das irgendwie nicht..SS: Ja, schwierig. In der Conclusions argumentieren wir ja nochmal, dass wir das in der Praxis nicht erwarten, dass die Wavefront zu komplex wird. Vielleicht reicht es da und wir können den Teil hier formal für sich stehen lassen. }

\section{Conclusions and Future Work}
We have identified and closed a seeming gap in literature
concerning a natural problem in computational geometry.
Namely, the question of whether there is a subcubic-time
algorithm that computes for a given polyline
an optimal simplification under the local \fre distance.
Simultaneous to us, Buchin et al.~\cite{buchin2022efficient}
have answered this question positively by
providing an $\oh(n^{5/2+\varepsilon})$ time algorithm as an application of a new data structure.
We have described an algorithm with a running time in $\oh(n^2 \log n)$,
which is worst-case optimal up to a logarithmic factor
for any algorithm that explicitly or implicitly constructs the whole shortcut graph.
Our algorithm is simpler and faster than the one of Buchin et al.

Our algorithm does not provide new techniques,
but we use, modify, and combine existing approaches
for similar polyline simplification algorithms,
which have not been described for this precise setting yet.
Although the Chan--Chin algorithm for the local \hau distance is asymptotically faster,
it requires two sweeps over the polyline to identify
the set of valid shortcuts,
while we maintain the elegance of the Melkman--O'Rourke algorithm
by requiring only a single sweep.
%Our algorithm is also faster than the algorithm by Guibas et al.\ for weak
%simplification, which requires rather complicated subroutines.
Moreover, our algorithm is simpler than the line stabbing algorithm by Guibas et~al.

As polyline simplification is a building block, for example,
for trajectory clustering algorithms \cite{brankovic2020k}
or simplification algorithms for more generalized structures \cite{bosch2021consistent},
our result may also help to improve running times of such approaches.
Due to its basic nature, there are many more use cases conceivable
where this algorithm can serve as a black-box subroutine.
Also, note that our running time is only a logarithmic factor slower than the widely used Douglas--Peucker heuristic for the local \fre distance \cite{Kreveld2020},
but we compute the optimal simplification.

Moreover, we conjecture that in practice, the running time of the algorithm we describe
should be %(close to)
quadratic even if one omits the tree data structure and
simply uses a linked list and linear searches to maintain and update the wavefront.
For a large wavefront to arise, the polyline vertices need to form a specific pattern, which is unlikely to occur naturally.
It would be interesting to validate this claim empirically~--
maybe even including the concept of $\nu$-light polylines.

Furthermore, the investigation of lower bounds could shed light
on the question of whether our upper bounds are tight.
Existing lower bounds only apply to simplification of polylines in high dimensions.
For the practically most relevant use case of two dimensions
no (conditional) lower bounds are known, though.
Another direction for future work would be to generalize
the algorithm to work in higher dimensions, which
%requires dealing with
implies
a more complex wavefront. 
Finally, one could also consider further distance measures,
as e.g.\ the \fre distance under the \lp metric for $p \in (0,1)$,
where the respective unit circles are not convex anymore,
which could make updating the wavefront data structure more expensive.

In some applications, it can be a drawback that in the classical definition of
a polyline simplification algorithm, the two endpoints of a polyline always need to be kept.
One could relax this constraint and investigate the setting that also the endpoints
may be removed as long as this removal does not violate the distance constraint.
This may be done in a (user) study using real-world examples.

\bigskip
\noindent\textbf{Acknowledgments.}
We thank Peter Sch\"afer for the helpful remark that
Case~\textsf{TB} and Case~\textsf{BT} in \cref{sec:wavefront-maintenance} cannot occur.

%%%%%%%%%%% References %%%%%%%%%%%
%\clearpage
\bibliography{polylines}

%%%%%%%%%%%  Appendix  %%%%%%%%%%%

\clearpage
\appendix

\section{Omitted Content from \cref{sec:algorithm}}
\label{app:algorithm}

\begin{figure}[t]
	\centering
	\includegraphics[page=1]{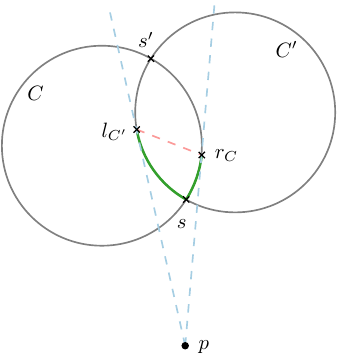}
	\caption{Illustration of the situation described in the proof of
		\cref{clm:bottom-arcs-dont-intersect-twice}.}
	\label{fig:bottom-arc-intersection}
\end{figure}

We start the appendix with a structural lemma, which we employ for the proofs of
\cref{lem:unit-circles-of-the-wave-front-contain-whole-wave-front,%
	lem:unit-circle-wave-front-two-intersections-from-below,%
	lem:update-wave-front-amortized-log,%
	lem:wave-front-has-at-most-two-intersection-points-with-another-circle}.
It does not yet use the wavefront.

\newcounter{saveTheorem}
\setcounter{saveTheorem}{\value{theorem}}
\setcounter{theorem}{\value{bottomArcsDontIntersectTwice}}
\bottomArcsDontIntersectTwice*
\label{clm:bottom-arcs-dont-intersect-twice*}
\setcounter{theorem}{\value{saveTheorem}}

\begin{proof}
	For an illustration of this proof see \cref{fig:bottom-arc-intersection}.
	Let $C$ and $C'$ be the two unit circles with the center of $C$
	being left of the center of  $C'$ w.r.t.~$p$.
	Now the cone between the right tangential from $p$ on $C$
	and the left tangential from $p$ on $C'$ contains all of the
	intersection area of $C$ and $C'$, and hence also both intersection points.
	We call the tangetial points $r_C$ and $l_{C'}$, respectively.
	Note that $r_C=l_{C'}$ is excluded as then $C$ and $C'$ would only
	have a single intersection point.
	For the intersection point $s$ between the bottom arcs of $C$ and $C'$,
	we know that the line segment $\overline{ps}$ does not intersect the
	inner part of any of the two circles by definition of the bottom arc.
	Hence the ray elongating this line segment has to go through
	the intersection area of $C$ and $C'$ above $s$.
	Therefore, the partial bottom arc of $C$ from $s$ to $r_C$
	and the partial bottom arc of $C'$ from $s$ to $l_{C'}$ are both
	on the boundary of the intersection area.
	As the intersection area is convex, it means that the
	line segment $\overline{l_{C'}r_C}$ is fully contained in the intersection area,
	and the intersection points have to be on opposite sites of
	the line through $l_{C'}$ and $r_C$.
	Accordingly, the second intersection point~$s'$ of $C$ and $C'$
	then has to lie above $\overline{l_{C'}r_C}$ and is therefore on the
	respective top arcs of $C$ and $C'$.
\end{proof}

We continue with another structural lemma,
which seems rather special at first glance, but we
employ it several times here in the appendix and in the main part,
e.g., in \cref{sec:wavefront-maintenance},
where we analyze the cases for the wavefront maintenance.

\begin{figure}[]
	\centering
	\begin{subfigure}[t]{0.3 \linewidth}
		\centering
		\includegraphics[page=1, trim={8 110 0 0},clip]{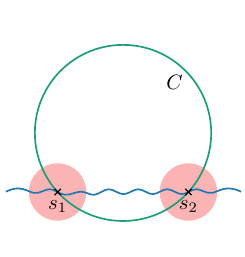}
		\caption{Situation that cannot occur.}
		\label{fig:unit-circle-intersects-wave-front-not-from-above}
	\end{subfigure}
	\hfill
	\begin{subfigure}[t]{0.31 \linewidth}
		\centering
		\includegraphics[page=2, trim={4 110 0 0},clip]{unit-circle-intersects-wave-front-not-from-above}
		\caption{$C$ intersects the wavefront twice with its bottom arc.}
		\label{fig:unit-circle-intersects-wave-front-not-from-above-Case-1}
	\end{subfigure}
	\hfill
	\begin{subfigure}[t]{0.32 \linewidth}
		\centering
		\includegraphics[page=3, trim={0 100 0 0},clip]{unit-circle-intersects-wave-front-not-from-above}
		\caption{$C$ intersects the wavefront with its top and bottom arc.}
		\label{fig:unit-circle-intersects-wave-front-not-from-above-Case-2}
	\end{subfigure}
	
	\caption{Sketch for \cref{lem:unit-circle-wave-front-two-intersections-from-below}:
		a unit circle~$C$ intersects the wavefront (blue wavy line) twice.}
\end{figure}

\setcounter{saveTheorem}{\value{theorem}}
\setcounter{theorem}{\value{twoIntersectionFromBelow}}
\twoIntersectionFromBelow*
\label{lem:unit-circle-wave-front-two-intersections-from-below*}
\setcounter{theorem}{\value{saveTheorem}}

\begin{proof}
	Clearly, if at $s_1$ the top arc of~$C$ intersects the wavefront,
	then on the left side of~$s_1$, $C$~is below the wavefront.
	Symmetrically, the same holds for~$s_2$.
	
	Now assume that at $s_1$ and at $s_2$, the bottom arc of~$C$
	intersects the arcs $a_j$ and $a_k$ of the wavefront, respectively.
	We denote their unit circles by $C_j$ and $C_k$.
	W.l.o.g.\ let $C$ on the left side of $s_1$ be above the wavefront.
	By \cref{lem:unit-circles-of-the-wave-front-contain-whole-wave-front},
	$C_j$ contains the rest of the wavefront including all of $a_k$.
	This means, that $C$ intersects $C_j$ at $s_3$ in between $s_1$ and $s_2$
	(potentially $s_2 = s_3$ if $C_j = C_k$); see
	\cref{fig:unit-circle-intersects-wave-front-not-from-above-Case-1}.
	Because the intersection of $C$ at $s_2$ is with the bottom arc of $C$,
	the intersection of $C$ and $C_j$ at $s_3$ is also with the bottom arc of $C$.
	This contradicts \cref{clm:bottom-arcs-dont-intersect-twice}.
	
	Finally, assume w.l.o.g.\ that at $s_1$ the top arc of~$C$
	intersects the arc $a_j$ of the wavefront
	and at $s_2$ the bottom arc of~$C$ intersects the arc~$a_k$ of the wavefront;
	see \cref{fig:unit-circle-intersects-wave-front-not-from-above-Case-2}.
	Again by \cref{lem:unit-circles-of-the-wave-front-contain-whole-wave-front},
	the unit circle~$C_j$ of $a_j$ contains the wavefront including the whole
	arc~$a_k$.
	Hence, there is an intersection point $s_3$ of $C$ and $C_j$ in between $s_1$ and $s_2$
	(where $s_2 \ne s_3$ and $C_j \ne C_k$ as otherwise $C$ and $C_j$ would
	have an intersection between their bottom arcs and between a bottom and a top arc).
	At $s_3$ there is the bottom arc of~$C$
	(since later at $s_2$, there is also the bottom arc of~$C$ involved).
	If $C_j$ also would have its bottom arc at $s_3$,
	it would contradict \cref{clm:bottom-arcs-dont-intersect-twice}.
	Therefore, at $s_3$, there is the top arc of $C_j$.
	This however means that $s_3$ is outside~$D_{i,j}$~-- a contradiction.
\end{proof}

\begin{figure}[]
	\centering
	\includegraphics[page=1, trim={0 82 0 0},clip]{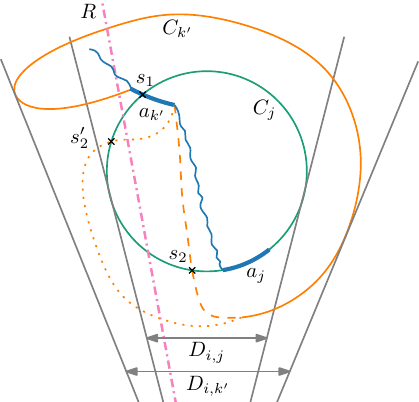}
	\caption{Configuration used to prove
		\cref{lem:unit-circles-of-the-wave-front-contain-whole-wave-front}.
		The blue part is the wavefront including the arcs~$a_{k'}$ and $a_j$.
		The circles~$C_j$ and $C_{k'}$ are unit circles and the gray rays
		indicate the local wedges.}
	\label{fig:wave-front-inside-unit-circle}
\end{figure}

\setcounter{saveTheorem}{\value{theorem}}
\setcounter{theorem}{\value{waveFrontInsideUnitCircle}}
\waveFrontInsideUnitCircle*
\label{lem:unit-circles-of-the-wave-front-contain-whole-wave-front*}
\setcounter{theorem}{\value{saveTheorem}}

\begin{proof}
	We argue that, for all $j \in \{i+1, \dots, k\}$,
	the claim is true by considering first all arcs that had been added before
	and then all arcs that had been added after
	the arc of $C_j$ had been added to the wavefront.
	
	All arcs~$a_{j'}$ on the wavefront belonging to a vertex $p_{j'}$ with $j' < j$
	are inside~$C_j$ because when the wavefront of~$W_{i,j}$ has been constructed,
	the wavefront of~$W_{i, j}$ consisted of arcs of the
	wave of $D_{i, j}$, i.e., arcs of~$C_j$,
	and it consisted of arcs of the wavefront of $W_{i, j-1}$ lying
	inside $I$, i.e., the intersection between $C_j$
	and the valid region of $W_{i, j-1}$%
	\footnote{We remark that even without the extra narrowing step using~$I$,
		if $C_j$ contributed an arc of the wavefront of~$W_{i, j}$,
		$C_j$ contained the whole wavefront of~$W_{i, j}$.}.
	
	All arcs~$a_{k'}$ on the wavefront belonging to a vertex $p_{k'}$
	with $j < k' \le k$ are completely inside~$C_j$ because if they were not,
	there would be an $a_{k'}$ (which is part of the bottom arc of the unit
	circle~$C_{k'}$) that intersects $C_j$ at~$s_1$;
	see \cref{fig:wave-front-inside-unit-circle}.
	The intersection at~$s_1$ is with the top arc of $C_j$ as otherwise
	$a_{k'}$ would be (partially) outside the local valid region of~$D_{i,j}$.
	Still for $a_j$ to be in the local valid region of~$D_{i,k'}$,
	$C_j$ and $C_{k'}$ intersect a second time.
	We consider two possible cases for a second intersection and
	denote them by $s_2$ and $s_2'$.
	First, assume that the intersection $s_2$ is between the bottom
	arc of $C_{k'}$ and the bottom arc of $C_j$.
	This however contradicts \cref{clm:bottom-arcs-dont-intersect-twice}
	because in $s_1$, there was already the bottom arc of $C_{k'}$ involved.
	Hence, the second intersection point is $s_2'$
	which is an intersection between the bottom
	arc of $C_{k'}$ and the top arc of $C_j$.
	Then, however, there is a ray~$R$ originating in $p_i$
	that lies in between $s_1$ and $s_2'$
	and intersects the bottom arc of $C_{k'}$ at least twice~--
	a contradiction.
\end{proof}

\section{Omitted Content from \cref{sec:wavefront}}
%\section{Omitted Content from \cref{sec:wavefront-size}}
\label{app:wavefront-size}

\setcounter{saveTheorem}{\value{theorem}}
\setcounter{theorem}{\value{atMostTwice}}
\atMostTwice*
\label{lem:wave-front-has-at-most-two-intersection-points-with-another-circle*}
\setcounter{theorem}{\value{saveTheorem}}
\begin{proof}
	We prove this statement inductively.
	Say $p_i$ is our start vertex and we consider the wavefront of $W_{i, i + 1}$,
	which is the same as the wave of $D_{i, i + 1}$, which is part of the boundary of a unit circle.
	Since each pair of unit circles in the \lp norm for $p \in [1, \infty]$ intersects
	at most twice, we know that any unit circle intersects the
	wavefront of $W_{i, i + 1}$ at most twice.
	
	It remains to show the induction step for all $j > i + 1$.
	Assume for a contradiction that a unit circle~$C$ intersects
	the wavefront of~$W_{i, j}$ more than twice.
	Observe that the wavefront of~$W_{i, j}$ is a subset of the wavefront of $W_{i, j - 1}$
	and the wave of~$D_{i, j}$.
	Say $C$ intersects the wavefront of~$W_{i, j - 1}$ at $q_1$ and $q_2$
	and $C$ intersects the wave of $D_{i, j}$ at $q_3$ and $q_4$
	(maybe one of these points does not exist.)
	Next, we argue topologically that at most two points of $\{q_1, q_2, q_3, q_4\}$
	lie on the wavefront of~$W_{i, j}$, which is a contradiction.

	\begin{figure}[]
		\centering
		\begin{subfigure}[t]{0.45 \linewidth}
			\centering
			\includegraphics[page=1]{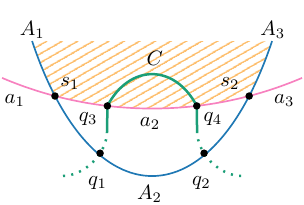}
			\caption{\centering Case A.}
			\label{fig:wave-front-and-unit-circle-intersect-at-most-twice-Case-A1}
		\end{subfigure}
		\hfill
		\begin{subfigure}[t]{0.45 \linewidth}
			\centering
			\includegraphics[page=2]{wave-front-and-unit-circle-intersect-at-most-twice}
			\caption{\centering Case B.}
			\label{fig:wave-front-and-unit-circle-intersect-at-most-twice-Case-A2}
		\end{subfigure}
		
		\caption{Cases in the proof of
			\cref{lem:wave-front-has-at-most-two-intersection-points-with-another-circle}.}
		\label{fig:wave-front-and-unit-circle-intersect-at-most-twice}
	\end{figure}
	
	By the induction hypothesis, the wavefront of~$W_{i, j - 1}$
	and the wave of~$D_{i, j}$ intersect at most twice.
	Let these intersection points from left to right be $s_1$ and $s_2$; see
	\cref{fig:wave-front-and-unit-circle-intersect-at-most-twice}.
	Let the subdivisions of the wavefront of~$W_{i, j - 1}$ and the wave of~$D_{i, j}$
	induced by~$s_1$ and~$s_2$ be $A_{1}, A_{2}, A_{3}$ and $a_{1}, a_{2}, a_{3}$, respectively.
	Some of them may be empty.
	Clearly, the wavefront of $W_{i, j}$ is either $A_{1}$--$a_{2}$--$A_{3}$ or $a_{1}$--$A_{2}$--$a_{3}$.
	By \cref{lem:unit-circle-wave-front-two-intersections-from-below},
	we know that it cannot be $a_{1}$--$A_{2}$--$a_{3}$,
	therefore, it is $A_{1}$--$a_{2}$--$A_{3}$.
	
	Next, we analyze the intersection points~$q_3$ and $q_4$
	(maybe $q_4$ does not exist).
	Either one or two of them lies on $a_{2}$ as otherwise there are no
	more than two intersection points of $C$ with the new wavefront.
	
	\medskip\noindent\textbf{Case~A:}
	The intersection points $q_3$ and $q_4$ lie on $a_{2}$;
	see \cref{fig:wave-front-and-unit-circle-intersect-at-most-twice-Case-A1}.
	As both intersection points are between the unit circle~$C$ and the wave
	of $D_{i,j}$, i.e., a bottom arc of another unit circle,
	we know by \cref{clm:bottom-arcs-dont-intersect-twice}
	that $q_3$ and $q_4$ are contained in the top arc of~$C$.
	Thus, there is no ray~$R$ to the left of $q_3$ or to the right of~$q_4$
	originating at~$p_i$ and intersecting the arc of~$C$ between~$q_3$ and~$q_4$
	as otherwise $R$ would intersect the top arc of~$C$ twice.
	Therefore, the arc of $C$ between $q_3$ and $q_4$ lies in the valid region
	(hatched orange in \cref{fig:wave-front-and-unit-circle-intersect-at-most-twice-Case-A1})
	without reaching $A_{1}$ or $A_{3}$.
	When $C$ passes through $q_3$ and $q_4$, it reaches the region between $a_{2}$ and~$A_{2}$.
	If there are intersections between $W_{i, j - 1}$ and $C$, they both lie on~$A_{2}$.
	
	\medskip\noindent\textbf{Case~B:}
	Only one intersection point, let it be $q_3$, lies on $a_{2}$;
	see \cref{fig:wave-front-and-unit-circle-intersect-at-most-twice-Case-A2}.
	If it is a touching point, then $C$ lies in the region
	between $a_{2}$ and~$A_{2}$ before and after reaching $q_3$
	(because we can assume that both unit circles are non-identical).
	If it is an intersection point, then $C$ passes through $q_3$ into
	the region between $a_{2}$ and~$A_{2}$ (hatched orange in
	\cref{fig:wave-front-and-unit-circle-intersect-at-most-twice-Case-A2}).
	To leave this region, $q_1$ (or $q_2$) lies on $A_{2}$.
	Hence, there are at most two points of $\{q_1, q_2, q_3, q_4\}$ on
	the new wavefront.
\end{proof}

\end{document}